%% file: main.tex
\documentclass[letterpaper]{article}
\usepackage{aaai2026} 
\usepackage{times} 
\usepackage{helvet} 
\usepackage{courier} 
\usepackage[hyphens]{url} 
\usepackage{graphicx}
\usepackage{natbib} 
\usepackage{caption}
\usepackage{algorithm}
\usepackage{algorithmic}
\usepackage{amsmath, amssymb, graphicx, amsthm, thmtools, mathtools}
\usepackage{subcaption}
\usepackage{placeins}
\usepackage{multirow,multicol}

\usepackage{cleveref}
\usepackage{thm-restate}

\newtheorem{definition}{Definition}

\newtheorem{example}{Example}
\usepackage{enumitem}
\usepackage{tikz}

\newcommand{\calX}{\mathcal{X}}

\newcommand{\choice}{\mathcal{C}}

\Crefformat{section}{\S#2#1#3}
\Crefmultiformat{section}{\S\S#2#1#3}{ and~#2#1#3}{, #2#1#3}{ and~#2#1#3}
\Crefformat{subsection}{\S#2#1#3}
\Crefformat{subsubsection}{\S#2#1#3}

\usepackage{color-edits}

\usepackage{xcolor}

\usepackage{amsmath}
\usepackage{newfloat}
\usepackage{listings}
\DeclareCaptionStyle{ruled}{labelfont=normalfont,labelsep=colon,strut=off}
\floatstyle{ruled}
\newfloat{listing}{tb}{lst}{}
\floatname{listing}{Listing}
\title{Capacity Constraints Make Admissions Processes Less Predictable}
\author {
    Evan Dong\textsuperscript{\rm 1},
    Nikhil Garg\textsuperscript{\rm 1},
    Sarah Dean\textsuperscript{\rm 1}
}
\affiliations {
    \textsuperscript{\rm 1}
    Cornell University\\
    edong@cs.cornell.edu, ngarg@cornell.edu, sdean@cornell.edu
}

\usepackage{bibentry}

\begin{document}

\maketitle

\begin{abstract}
Machine learning models are often used to make predictions about admissions process outcomes, such as for colleges or jobs.  However, such decision processes differ substantially from the conventional machine learning paradigm. Because admissions decisions are \textit{capacity-constrained}, whether a student is admitted depends on the other applicants who apply. We show how this dependence affects predictive performance even in otherwise ideal settings. Theoretically, we introduce two concepts that characterize the relationship between admission function properties, machine learning representation, and generalization to applicant pool distribution shifts: \textit{instability}, which measures how many existing decisions can change when a single new applicant is introduced; and \textit{variability}, which measures the number of unique students whose decisions can change. Empirically, we illustrate our theory on individual-level admissions data from the New York City high school matching system, showing that machine learning performance degrades as the applicant pool increasingly differs from the training data. Furthermore, there are larger performance drops for schools using decision rules that are more unstable and variable. Our work raises questions about the reliability of predicting individual admissions probabilities.
\end{abstract}

\section{Introduction}

It is common to use statistical models based on historical data to make predictions about the outcomes of applications to selective programs.
Such models are often used by decision makers, as in the case of recruiters using resume filters or college admissions officers ranking candidates for review.
They are also used by applicants to understand their own admissions chances, influencing application decisions.
The use of machine learning (ML) in these settings is subject to substantial scrutiny, stemming from a combination of ethical concerns and from the limitations of ML, such as regarding bias, data limitations, and lack of transparency.
Admissions decisions are high stakes, so errors have significant consequences; on the other hand, computational tools may be useful for both applicants and decision makers.

In this paper, we identify and analyze a new source of trouble for machine learning as classically used in the context of admissions.
Admissions are inherently capacity constrained---there are only so many seats in class.
As a result, the decision to admit one applicant inherently affects the chances of others in the applicant pool. 
This phenomenon of \emph{cohort-dependence} clashes with the independence assumptions and formulation of machine learning.
We show that the extent of this dissonance relates to properties of the decision process, and may be worse for programs that seek to admit diverse or balanced cohorts of applicants. In short, historical application decisions -- that compose the training set -- were functions of historical applicant pools, in ways that are not explicitly modeled in traditional machine learning. As we show, then, even small applicant pool changes can change true outcomes -- applicants who would have truly been accepted historically may be rejected in the future, and vice versa. Our contributions are as follows. 

\textbf{(1)}
We introduce formal properties of admission decisions which affect the difficulty of generalization (\Cref{sec:theory-definitions});
\textbf{(2)} we present a characterization which bridges the theory of choice functions with machine learning (\Cref{sec:theory-results});
\textbf{(3)} We illustrate our theory using admissions data from the New York City high school matching system (\Cref{sec:empirics}).
Before diving in, we begin by reviewing related work (\Cref{sec:related-work}) and introducing the problem setting, notation, and background (\Cref{sec:setting}). We conclude with a discussion of implications in \Cref{sec:discussion}.

\section{Related Work} \label{sec:related-work}

\paragraph{Machine Learning for Admissions}

Many works apply statistical models and machine learning to undergraduate admissions \citep{bruggink1996statistical,lux2016applications,neda2022feasibility,kiaghadi2023university,lee2023augmenting}, and similar lines of work for graduate admissions \citep{moore1998expert,waters2014grade,gupta2016will,staudaher2020predicting} and medical residency \citep{rees2023machine}. For example, see \citet{lee2024ending} and the references therein.
Many of these works primarily demonstrate the possibility of modeling admissions outcomes with machine learning.
Another motivation is auditing for bias, where statistical models of admission probability have played an important role, for example in the Supreme Court case \emph{Students for Fair Admissions v. Harvard} 
\citep{sffaamicus}.

In other cases, statistical models are geared towards deployment, though generally not to automate the entire decision process.
\citet{lee2023evaluating} describe a highly-selective college in the United States that uses a predictive model to provide coarse estimates that are used to group and rank the order of students' applications for human review.
Models are also deployed to give guidance to the applicants, for example, to universities \citep{sirolly2024impact,collegevine}. 
 
In the context of our empirical illustration, the New York City Department of Education deployed a predictive model in the 2024-2025 application year to assist students in assessing their odds of admission to high schools \cite{ShenBerro2024HighSchoolAdmissions}.
It is important to note that the model -- which displays coarse outcome predictions (low, medium, high) -- does \textit{not} use machine learning, but rather directly simulates the (known) choice functions in use by schools. Importantly, this approach does \textit{not} face the challenges we characterize in this work, since simulating the choice functions does take into account other applicants, for various assumptions on what the applicant pool would be. Such an approach is one path forward for prediction in admissions settings as opposed to standard machine learning approaches directly, though may not be feasible in settings where the choice functions are not directly known and can be simulated -- for example, in this context, some schools use covariates unavailable to researchers, and so some machine learning component may be necessary. Our work evaluates the challenges of such a potential approach.

\paragraph{Pitfalls of Social Prediction}

Prediction of social outcomes, of which admissions is one example, can be fraught.
In the context of education, \citet{perdomo2025difficult} argue that personalized predictions of high school graduation in Wisconsin do not outperform high school-level predictions and interventions.
\citet{liu2023reimagining} discuss the mismatch between supervised machine learning problem formulation and actual educator needs. \citet{raji2022fallacy} argue that AI models are often erroneously assumed to be effective, when in reality they simply do not work as advertised. \citet{wang2024against} provide a list of reasons why predictive machine learning modeling may fail within their normative argument against `predictive optimization' -- for example, distribution shifts induce fundamental limits to prediction. To this literature, we contribute a novel \textit{explanatory mechanism} driving poor predictive performance in admissions settings: cohort-dependence induced distribution shifts. 

A large literature benchmarks (poor) performance of machine learning algorithms under distribution shift, and aims to develop distributionally robust algorithms \citep{gardner2023benchmarking,koh2021wilds,yao2022wild}. Such approaches often (implicitly or explicitly) specify a model of how distributions can shift at test time \cite{kaurmodeling}. To our knowledge, these approaches stay within the ML paradigm of modeling decisions as independent; our results further suggest that developing robust predictors in ML settings requires modeling how decisions depend on entire applicant pools.

Other recent works also study independence in decision-making in machine learning model. Most related, \citet{dongdiscretization2025} challenges independent decision-making -- in the context of discretizing continuous scores for demographic imputation (e.g., race prediction); both works suggest that decisions for an individual data point should depend on the entire cohort (sample), but the rationales differ: in this work, the \textit{true} decisions for an applicant depend on inference-time application pools; in \citet{dongdiscretization2025}, the true labels (e.g., the demographics of an individual) are fixed, but independent decision-making biases the imputed label distribution and downstream tasks.

Finally, we note that a recent line of work studies the pitfalls (and potential benefits) of algorithmic monoculture, in which multiple machine learning models make correlated decisions \textit{for the same data point} \cite{kleinberg2021algorithmic,bommasani2022pickingpersondoesalgorithmic, creel2022algorithmic, toups2023ecosystem, jain2023algorithmic,peng2024monoculture,peng2024wisdomfoolishnessnoisymatching, jain2024scarce,kim2025correlated}; in contrast, this work studies the dependence of true labels \textit{across} inference-time datapoints.  

\paragraph{Formal Models of Choice and Admission}
The formalization of admissions processes ties into a long and broad economic literature on choice functions, preferences, and behavior \citep{arrow1959rational,sen1971choice,moulin1985choice,kalai2002rationalizing}. The basic structure of choice functions have formed the basis for voting theory, utility theory, matching markets, and more.
Much of the seminal choice function literature focuses on narrowing the vast space of possible choice functions by suggesting desirable properties and proving relations between them.

A related literature on affirmative action studies choice functions subject to diversity constraints---often assuming a given ordering of preferences over students \citep{echenique2015control,celebi2023diversity,arnosti2024explainable}. Much of this literature focuses on improving the definition of constraints or designing better mechanisms for implementing preferences subject to the constraints. 
In comparison, our work treats the choice function as a given part of the data generating process, and we explore questions of representation via machine learning. We relate representation ability to classic properties in economic models, such as substitutability. 

\section{Model and Setting} \label{sec:setting}

\subsection{Machine Learning}

Data from admissions processes are often used for supervised machine learning as follows:  each applicant is a data point and admissions decisions are labels. 
Applicants are represented by the features $x\in\mathcal X$ on their application, and labels $y\in\{0,1\}$ are binary (accept or reject).
Then the relevant task is binary classification, i.e., generating a model
${f}: \calX \to \{0, 1\}$ using data $X = \{x_1, x_2, ..., x_n\}$ and $Y = \{y_1, y_2, ..., y_n\}$ from past admissions processes.

Training a machine learning model means selecting a $f$ from a set of possible models, i.e., a hypothesis class.
Generally, this is done by finding a model which achieves low error on available data (a train set).
It is typical to use continuous optimization algorithms to first train a continuous \emph{score} function $s:\mathcal X\to [0,1]$ by minimizing the error. 
Then the classification model is defined as 
\begin{align}\label{eq:ml-independent}
    f(x) = \mathbf 1\{s(x)\geq t\}
\end{align} based on a fixed threshold value $t$, often $t=0.5$. 
Notice that models of the form~\eqref{eq:ml-independent} predict outcomes independently for each applicant.
At test time inference, they do not consider the overall applicant pool -- they implicitly model the training time pool, through its effect on the training labels.

\subsection{Choice Functions}

Many admissions processes have capacity constraints (either perceived or actual). This fact gives rise to what we call \textit{cohort-dependence}: an applicant's acceptance decision depends on not only their own characteristics, but also the number and characteristics of other applicants.

\textit{Choice functions} formalize admission decision processes and allow modeling such dependence. 
A choice function describes the behavior of ``choosing'' elements from a set of available options, and is commonly used in economic models. Formally, a choice function over the universe $\calX$ is a function $\choice: 2^\calX \to 2^\calX$ that maps all subsets $X$ of $\calX$ to the accepted set $\choice(X)$, where $\choice(X) \subseteq X$. As our motivation focuses on real applicant datasets, we consider finite subsets. 

Recalling admissions data, the binary admissions label $y$ for student $x$ is 1 if $x \in \choice(X)$ and 0 otherwise.
To aid in our discussion, we will slightly abuse notation to relate pairs $(x_i,y_i)$ to the choice function $\mathcal C(X)$ where $x_i\in X$. 
Define $\mathcal C(X)_i$ as 1 if $x_i\in \mathcal C(X)$ and 0 otherwise; this allows us to write the label $y_i=\mathcal C(X)_i$.
Notice that the dependence of labels on the applicant pool introduces a particular form of distribution shift called ``concept drift,'' where the conditional distribution $y | x$ changes over time.\footnote{This contrasts with ``covariate shift'', where new datasets contain shift $x$ distributions, but maintain the same relationship $y | x$.} This occurs even when the choice function itself is held constant.

A basic property of choice functions we consider is that they are capacity constrained to accept at most $q$ applicants.

\begin{restatable}[$q$-Acceptance]{definition}{qacceptance} 
\label{def:q-acceptance}
Choice function $\choice$ is $q$-acceptant if $|\choice(X)| = \min\{q, |X|\}$ for all input sets $X$.
\end{restatable}
\noindent $q$-acceptant choice functions accept as many applicants as possible, up to $q$ elements.
Some (but not all) choice functions can be characterized as a total order. 

\begin{definition}[Total order] \label{def:total_order}
    A \emph{total order} is a transitive, asymmetric, and complete binary relation $\succ$ over elements of $\calX$; that is, $x_1 \succ x_2$ implies $x_2 \not\succ x_1$, $x_1 \succ x_2$ and $x_2 \succ x_3$ implies $x_1 \succ x_3$, and all $x_1, x_2$ are comparable.
\end{definition}

Total orders are frequently used in modeling choice mechanisms; they are often called \textit{preference orderings} or \textit{priority orders} in economic models. 
Intuitively, total orders induce a ranking over candidates and could result from, e.g., ordering applicants based on their GPA.
We will say that a $q$-acceptant choice function $\mathcal C$ is \textit{characterized by a total order} if there exists a total order $\succ$ such that for any input $X$, the accepted applicants are the top-$q$ elements of $X$ ordered by $\succ$.
In other words,  $x_1 \succ x_2$ for all $x_1 \in \choice(X), x_2 \not\in\choice(X)$. 
We will refer to the process of sorting elements according to a total order $\succ$ and selecting the top $q$ as a \textit{queue}. 
Not all choice functions are characterized by total orders; consider a program which admits the $\frac{q}{2}$ applicants with the highest English scores and $\frac{q}{2}$ of the rest with the highest Math scores.

\section{Theoretical Analysis}\label{sec:theory}

We now present a theoretical framework for understanding why admission decisions are hard to predict with ML models, due to capacity constraint-induced dependence. We define two choice function properties: \emph{instability} (the true labels of multiple data points can change when a new data point is added) and \emph{variability} (\textit{which} labels change can depend on the new data point), and characterize choice functions according to these properties. Proofs are in \Cref{sec:proofs-appendix}.

\subsection{Instability and Variability} \label{sec:theory-definitions}

To characterize the instability of a choice function, we will measure the number of decisions which change as the applicant pool expands. 
To that end, we first define a distance between sets of applicants induced by the choice function.

\begin{restatable}[Choice Distance]{definition}{distance} \label{def:distance}
For a choice function $\choice$ and sets $X_1 \subseteq X_2$, we define the choice distance as:
\begin{equation}
r_\choice(X_1, X_2) := |X_1 \cap \choice(X_2) \setminus \choice(X_1)| + |\choice(X_1) \setminus \choice(X_2)|
\end{equation}
\end{restatable}

This distance\footnote{We use the term ``distance'' colloquially here; this definition does not fulfill the criteria to be a proper distance metric.} measures the number of decisions that \textit{change} when expanding the candidate pool from $X_1$ to $X_2$. The first term counts rejections that become acceptances; the second counts acceptances that become rejections.

\begin{restatable}[$d$-Instability]{definition}{defstability} 
\label{def:d-instability}
Choice function $\choice$ is $d$-unstable if for all $X$, and $X' = X \cup \{x'\}$, we have $r_\choice(X, X') \leq d$ for any nonnegative integer $d$. $\choice$ is called tightly $d$-unstable if it is $d$-unstable but not $(d-1)$-unstable.
\end{restatable}

Intuitively similar to a Lipschitz bound, note that because $r_\choice$ obeys the triangle inequality (\Cref{thm:triangle-ineq}), $d$-instability also guarantees at most $dk$ changes when adding $k$ elements.
This definition focuses on changes only to \textit{existing applicants'} decisions. Unlike in conventional statistical learning, this quantity does not relate to performance on new data points, but rather, isolates the effect of new data points on \textit{existing} data points---an effect not traditionally considered.

While instability characterizes how many chosen elements change when perturbing the input set, we also consider \textit{which} elements change.
Even for a 1-unstable choice function, there is still substantial diversity in \textit{which} decision might change. 
For example, in a school that simply selects the $q$ students with the highest GPA, the same ``borderline'' applicant is displaced regardless of who the new (higher GPA) applicant is. 
In contrast, a school that selects half of their students based on an English test score and the other half based on Math would reject different students depending whether the new applicant excels at English or at Math.

We formalize\footnote{For a more general definition of variability (beyond 1-unstable and $q$-acceptant), see \Cref{sec:variability-appendix}.} this intuition with the following definition.

\begin{restatable}[Variability]{definition}{defvariability} 
\label{def:variability}
A $1$-unstable, $q$-acceptant choice function $\choice$ has variability $m$ where
\begin{equation}
m \coloneqq \max_{X\subset \calX} \left|\bigcup\limits_{x'\in\calX} \choice(X) \setminus \choice(X \cup \{x'\})\right|
\end{equation}
\end{restatable}
\noindent Variability bounds how many \textit{different} currently accepted candidates could be displaced by adding any single new candidate. As we empirically demonstrate, these notions capture how sensitive decisions are due to applicant pool shifts. 

\subsection{Main Theoretical Results} \label{sec:theory-results}

We first relate the concepts of instability and variability to the practice of machine learning.
We focus on the representation capabilities of ML models: when is it possible for an ML model to faithfully represent an admissions process?
Formally, we say that an ML model can represent a choice function $\mathcal C$ if there exists a function $f$ such that for all applicant sets $X$, $f(x_i) = \mathcal C(X)_i$.
{Note that by focusing on representation, our theory is distribution agnostic, and applies to settings where applicant distributions are affected considerations such as strategic behavior and matching algorithms.}

Before stating this result, consider a straightforward way to adapt a model~\eqref{eq:ml-independent} to the capacity-constrained setting.
Given that the ML model is based on continuous scores, we can adjust the decisions to account for a known capacity constraint if we have available the set of all test-time applications.
In particular, we can rank the applicants by their scores and then select the top $q$.
This process results in decisions determined by a cohort dependent threshold $t_q(X)$ which is the score of the $q$th ranked applicant:
\begin{align}\label{eq:ml-rank}
    f(x; X) = \mathbf 1\{s(x) \geq t_q(X)\}.
\end{align}

\begin{restatable}[ML Representation]{proposition}{mlrepresentation} 
\label{prop:ml-representation}
A model of the form~\eqref{eq:ml-independent} which makes independent predictions can only represent 0-unstable choice functions.
A model of the form~\eqref{eq:ml-rank} which ranks applicants can represent a 1-variable, 1-unstable choice function. No such model can represent a choice function with variability or instability greater than one.
\end{restatable}

This result shows that standard ML practice coherently applies to admissions data only in limited settings.
The intuition for this result hinges on the fact that all 1-unstable and 1-variable choice functions correspond to total preference orderings over applicants (formalized in Theorem~\ref{thm:variability}).
Then the proposition follows from noticing that total preference orderings correspond to  embedding applicants in  $\mathbb{R}$  (by their scores).
How restrictive are the settings where naive ML applies? We next characterize the instability of choice functions.
First, we introduce a choice function property which plays an important role in the school admission setting.

\begin{restatable}[Substitutability]{definition}{substitutability} 
\label{def:substitutability}
Choice function $\choice$ is substitutable if $X_1 \subseteq X_2$ implies $X_1 \cap \choice(X_2) \subseteq \choice(X_1)$.
\end{restatable}

Substitutability means that removing other applicants from the pool cannot hurt an accepted applicant. Substitutability is well-established property in the choice literature \citep{chernoff1954rational,moulin1985choice,deng2017complexity}. It is a necessary condition for the existence of a stable matching \citep{roth1984stability}. Many common optimization objectives are substitutable; selecting applicants by a linear ranking,  linear assignment problems, and selection problems optimally solved by greedy algorithms \citep{yokoi2019matroidal} can all be represented as substitutable choice functions.

\begin{restatable}{theorem}{thminstability}
\label{thm:instability}
A $q$-acceptant choice function 
\begin{enumerate}
    \item cannot be $0$-unstable,
    \item  is exactly $1$-unstable if and only if it is substitutable,
    \item can be tightly $d$-unstable for every $1\leq d \leq 2q$.
\end{enumerate}

\end{restatable}%

First, capacity constraints directly imply nonzero instability; a school will become more selective as the applicant pool widens.
As a result, recalling \Cref{prop:ml-representation}, naive ML of the form~\eqref{eq:ml-independent} cannot represent such decision processes.
Capacity constraints force decisions to be cohort-dependent, violating traditional independence assumptions. 

Skipping to the third point, we show that there exist choice functions which are much more unstable.
One intuitive class of highly unstable (large $d$) mechanisms is team formation. For example, a music department may want a balanced makeup of instrumentalists accepted to a program; the music director prefers equal numbers of violinists and cellists, but otherwise chooses the most talented musicians. Due to a lack of cellists, after admitting all auditioning cellists, there remain several violinists rejected, and additional slots are filled with, e.g., bassists. Adding an additional cellist to the applicant pool, however, leads to rejecting two bassists in favor of a previously rejected violinist and the new cellist---a $3$-unstable function. 
If admissions instead prioritized string trios (a violin, viola, and cello), there might be \textit{three} bassists rejected to form a trio ($5$-unstable), while forming a \textit{quartet}, (a viola, cello, and two violins), is $7$-unstable. We show in \Cref{lem:unstable-calc} that complementary groups of size $n$ can define functions that are tightly $2n-1$ unstable.
We also construct examples of instability $2n$ when additional applicants alter decisions without being accepted themselves in \Cref{sec:inconsistent-proof}.

Going back to the second point of Theorem~\ref{thm:instability}, we show that substitutability is equivalent to $1$-instability under $q$-acceptance---violations of substitutability \textit{must} lead to larger changes in admissions decisions, and vice versa. In other words, a broad class of functions are as unstable as possible, given that they are subject to capacity constraints. 

Whether ML models of the form~\eqref{eq:ml-rank} -- which rank continuous scores -- can represent $1$-unstable choice functions depends on their variability.
Before turning to this characterization, we return to the idea of orderings and queues.
Not all choice functions are characterized by a total order. 
However, a larger class of substitutable and $q$-acceptant choice functions can be defined as a combination of \textit{multiple} queues. 
The simplest way of constructing this is by sequentially composing several choice functions together.
\begin{restatable}[Sequential Composition]{definition}{seqcomposition} \label{def:seq-composition}
A choice function $\choice$ is a \textit{sequential composition} of functions $\choice_n, \ldots, \choice_1$ if:
$$\choice(X) = \choice_1(X) \cup \choice_{2}(X_{2}) \cup \cdots \cup \choice_n(X_n)$$
where $X_{i+1}= X_{i} \setminus \choice_{i}(X_{i})$ and $X_1 = X$.
\end{restatable}

For example, a school may first accept $q/2$ students with the highest English scores, and then $q/2$ students with the highest Math scores (who were not already admitted).

\begin{restatable}{theorem}{thmvariability} 
\label{thm:variability}
Consider a $q$-acceptant, 1-unstable choice function $\mathcal C$ which
can be represented as the sequential composition of $n$ choice functions, each characterized by a total order.
Then $\mathcal C$ has a variability  $m$, where $1 \leq m \leq n$. Furthermore, $\mathcal C$ has variability $1$ if and only if it can be characterized by a single total order, $n=1$.
\end{restatable}
In the appendix, we discuss when this upper bound is tight. We also extend this result to choice functions which are not explicitly defined as a fixed sequential composition, such as linear assignment problems.
\Cref{thm:lap-variability} addresses the context-dependent sequential queues which result from such optimization based admissions. 
Importantly, Theorem~\ref{thm:variability} means that a choice function can be faithfully represented by a machine learning model if and \emph{only if} it is characterized by a total ordering over $\calX$.

\section{Application: NYC High School Admissions}\label{sec:empirics}

We illustrate our theory on data from the NYC high school matching system. Using admissions and applicant data from the 2021-2022 and 2022-2023 admissions cycles,\footnote{We received the data from the New York City Board of Education through a research data use agreement process, to study application behavior and behavioral interventions in the application process. This data is available to researchers with a sponsor inside the Department of Education. Due to a non-disclosure agreement and to protect private student data, we cannot release the data. The research was deemed exempt by our university's Institutional Review Board.} we extract applicant pools and admissions for each individual program.
As we cannot release the data, we further generate synthetic data to validate our results, and release our code. The code is available at \url{https://github.com/evan-dong/admissions-prediction}, and the replication results with synthetic data are in the appendix. 

The motivation for this empirical demonstration is twofold. First, we seek to understand the challenges that may arise when designing interventions to help applicants make decisions; for example, New York City now provides (coarse) personalized admission likelihoods for each program to applicants \cite{ShenBerro2024HighSchoolAdmissions}. As noted above, their approach avoids the challenges discussed here by directly modeling the choice functions. However, such a simulation approach may not be feasible in settings where the choice function is not precisely known, such as in holistic college admissions or hiring, or for schools that use features other than numeric ones, such as essays or auditions. Thus, second, this demonstration serves as an illustration of the challenges of prediction in such settings broadly -- in a setting where, due to the ability to simulate outcomes, we can control for other reasons that such prediction may be difficult.

In this context, each student can apply to multiple programs and rank them in order of preference; in our data, they were limited to listing twelve programs. They are then matched to a single program according to the deferred acceptance algorithm. Programs admit students according to various criteria; for example, some schools admit students according to a musical audition or essays. The vast majority of programs use explicit and publicly available decision rules, for which we can construct counterfactual decisions according to alternative applicant pools;\footnote{The DOE publishes these criteria online \citep{nycdoe2022hsdata}.} we use a ``simulator'' developed by \citet{peng2025undermatching}, which implements the choice funtions used by a large set of programs. As in their work, we restrict our analysis to the most common categories of program admissions choice functions  that we can replicate accurately based on the data. 

In this section, we first apply our theoretical framework to choice functions used by NYC programs in practice, characterizing their instability and variability. Then, we show that these characteristics correspond to machine learning performance in predicting admissions outcomes over time. 

\subsection{Applying the Theoretical Framework}

We first describe the most common categories of admissions functions used by NYC High School programs and then characterize them based on our instability and variability definitions. We analyze three types of programs: ``Ed. Opt'', ``Screened,'' and ``Open.'' Programs may also participate in the Diversity In Admissions (DIA) initiative. In the language of our model, ``Screened'' and ``Open'' programs use the same choice function class. Thus, we have four possible categories of functions: ``Ed. Opt'', ``Screened/Open,'' ``Ed. Opt with DIA'', ``Screened/Open with DIA.''

At a high level, each function is a sequential composition of queues.
Within each queue, applicants may be ranked according to a queue-specific score function, Ed. Opt. category, DIA-qualifying status, borough (neighborhood) of middle school attendance and residence, (discretized) grade tier, continuing student status, and a lottery (tiebreaker) number. Each queue also has a capacity (which we take to be fixed).
Thus each queue corresponds to a choice function that is characterized by a total ordering.
The queues themselves are ordered, and so admissions follows a sequential composition of these choice functions (Definition~\ref{def:seq-composition}).

In a \textbf{Screened/Open} program, there is a single queue (ranking), with students ordered based on a score\footnote{Open programs differ from Screened programs by defining a score which does not depend on grades.} determined by a combination of their grades, neighborhood and continuing-student priorities, and the lottery number. In a \textbf{Screened/Open with DIA} program, there are two queues; in one, all students with DIA-qualifying status are ranked above the remaining students (and then by the score function); in the other queue, all students are ranked by the score function. In an \textbf{Ed. Opt.} program, there are three queues, respectively corresponding to whether a student had \textit{low}, \textit{medium}, or \textit{high} grades in middle school\footnote{The intention is to induce diversity in terms of academic levels.} -- students with these grades are ranked at the top in their respective queues (and otherwise by a score determined by their lottery number and neighborhood and continuing-student priorities). Finally, in \textbf{Ed. Opt with DIA}, there are six queues, each corresponding to a Ed Opt. grade tier plus DIA-qualifying status. We can thus apply our theoretical framework as follows.

\begin{restatable}[]{proposition}{empiricalmethods} 
\label{prop:program-variability}
All considered choice functions ({Ed. Opt}, {Screened/Open}, {Ed. Opt with DIA}, and {Screened/Open with DIA}) are $1$-unstable. Furthermore, their variability is (tightly) equal to their number of queues: 

\begin{itemize}[noitemsep,topsep=0pt,parsep=0pt,partopsep=0pt]
    \item Screened/Open programs are $1$-variable
    \item Screened/Open with DIA programs are $2$-variable
    \item Ed. Opt. programs are $3$-variable
    \item Ed. Opt with DIA programs are $6$-variable
\end{itemize}

\end{restatable}

\subsection{Empirical Methodology}

Our goal is to illustrate that capacity constraints---and shifts in the applicant pools over time---contribute to the difficulty of predicting application outcomes. For this, we require a dataset in which we have outcomes for applicants from the same choice functions but with different applicant pools. It is also important to isolate the effect of shifting applicant compositions from other reasons that application outcomes may be hard to predict.
These include changes to the decision process and choice function complexity more broadly (e.g. nonlinearity or high dimensionality). Finally, we want a range of choice functions, including those  with different instability and variability. In other words, we require, for each school program: a consistent \textit{choice function} used on different \textit{application pools} over time, with \textit{admissions outcomes for each applicants} given the pool. Then, we can train \textit{machine learning models} on outcomes from one application pool, and evaluate their performance on outcomes from different pools. We note that such data is in general difficult to obtain. We now describe each component.

At a high level, our empirical approach is designed to make prediction as \textit{easy} as possible---to remove all reasons that prediction may be difficult, other than the applicant pool composition effect that is the focus of this work; as we illustrate, prediction is still difficult over time for choice functions with poor instability or variability characteristics.

\paragraph{Programs, applicants, and applicant pools} 
We use data from the General Education match in the 2021 and 2022 New York City High School match process, comprising of 58,500 students in 2021 and 57,331 students in 2022. We consider their applications to 199 unique programs: 99 that are {Ed. Opt}, 75 that are {Screened/Open}, 5 that are {Ed. Opt with DIA}, and 20 that are {Screened/Open with DIA}.
\footnote{We start with $n=587$ programs that are either Ed. Opt $398$, Screened $129$, or Open ($60$) and existed in both 2022 and 2021. From there, we further restrict our dataset to programs that matched  with at least one student in each year ($n=573$) with a nontrivial offer rate (i.e., rejected at least one student in each year) ($n=199$).} 

The matching mechanism uses the ``deferred acceptance'' algorithm: students submit rankings over programs and programs' choices are determined according to the above defined policies. In short, the algorithm iteratively selects a student who is not tentatively accepted to a program, and that student ``applies'' to the top program in their list to which they have not yet been rejected. Given its current set (its current tentatively accepted students and the new applicant), the program may tentatively accept or (permanently) reject the new applicant, and may (permanently) reject a current tentatively accepted applicant. This algorithm proceeds until all students are tentatively accepted somewhere or have exhausted their submitted preferences. This algorithm means that a student truly ``applied'' to and was accepted by the program they were matched to, ``applied'' to and was rejected by any program they ranked above their match, and have indeterminate status at the remaining programs, since they never actually ``applied'' in the deferred acceptance sense. 

For our analysis, we do not consider the deferred acceptance algorithm. Rather, we consider each program separately, as making admissions decisions according to a choice function and given an applicant pool. For each program, we consider as the potential pool those students who truly applied (\textit{in the deferred acceptance sense}) to that program in either 2021 or 2022.\footnote{This can roughly be viewed as understanding the decisions within one step of the deferred acceptance algorithm, in which an applicant ``proposes'' to a school, which either rejects them immediately or tentatively adds them to their acceptance pool, given the current acceptance pool.} Based on the above process, we have a total of 107,128 and 104,501 applications for 2021 and 2022 respectively from the chosen students to the given programs.

\paragraph{Choice functions and outcomes}
We now describe how we construct outcomes for each applicant to a program given a pool. We adapt a \textit{simulator} developed by \citet{peng2025undermatching} to generate admissions decisions: for each program in each year, the rule-based simulator implements the admissions function defined by publicly available policies.
Then, given an application pool and a fixed capacity  (which we infer from actual outcomes), we can calculate outcomes for each student. This simulator is implemented using 2022 decision policies and accurately\footnote{Accuracy is not 100\% because there are additional criteria that apply to a small number of students and for which we do not have data (e.g., children of teachers or siblings of current students).} reflects actual admission outcomes when applied to the application pools and outcomes from both 2022 (98.79\% accuracy) and 2021 (91.47\% accuracy).

Why don't we simply use the actual admissions outcomes of each applicant to a program in 2021 and 2022, since we use the \textit{real} applicant pool as data and have access to these outcomes? First, as described above, real-world admissions functions may change from year to year, and our interest here is in studying inherent difficulties of learning admissions functions even when they remain the same. Second, we wish to also display the effect of small changes to the applicant pool (such as substituting a few students), instead of the full changes between years. Third, this approach allows evaluating the predictability of new choice functions, not currently used by any programs (including those that are not 1-unstable), under real applicant pool shifts over time. Similarly, we can apply the same choice functions to all programs, thus controlling for features (such as overall admission rates) that correlate with both prediction difficulty and the choice function a program actually used.\footnote{For example, if program A originally admitted students based on a screened admissions method, we can simulate the admissions outcomes had program A instead been an Ed. Opt. or Open program, with or without DIA-reserved seats. This allows us to deconfound choice function predictability from other correlations between the applicant pool and the method---for example, many of the most selective programs are screened programs with DIA.} 

In other words, a simulator allows us to generate counterfactual decisions---we can modify the applicant pool or the underlying choice function. It further allows us to control for other factors (aside from our focus on shifts in the  applicant pool composition) that may make predicting admissions outcomes difficult, cf. \citet{wang2024against}. By generating every combination of admissions method with every program (for a total of six; three admissions methods, with or without 50\% of seats being reserved for the DIA initiative), we can examine the effect of different admissions methods. We also isolate the impact of instability- and variability-related differences in methods in \Cref{sec:nopriority-appendix} by simulating admissions without borough or continuing student priority.
We always use the real applicant pools to each program in each year.
This allows us to capture the effects of real shifts in applicant pool composition on model accuracy. 

\textit{Synthetic Admissions Methods.} 
Notably, all choice functions from the programs we analyze are $1$-unstable. To further test the theory, we design $0$- and $5$-unstable methods that we implement in our simulator. We create two different $0$-unstable admissions methods not subject to capacity constraints; the first simply admits every student with a tiebreaker number below fixed threshold (calibrated so that approximately half the overall applicants will be admitted).The second similarly thresholds tiebreaker number, but with three different values depending on an applicant's Ed. Opt. category. Our $5$-unstable admission method  is detailed in \Cref{sec:synthetics-appendix}.

\subsubsection{Machine learning predictions of admissions outcomes}
The above processing results in a set of programs, real applicant pools in each year for each program, and admissions outcomes under each choice functions and given each applicant pool. We now describe our ML training and evaluation.

\textit{Training and test sets.} We aim to replicate real-world ML performance and training: when a model is trained on admissions outcomes from one time period, and then used to predict outcomes in a future time period. Crucially, all admissions outcomes in the training data are based on the train-time applicant pool, while the machine learning model will be used to predict outcomes for students given a future applicant pool. For example, we may train a model on 2021 applicant outcomes, to give advice to applicants as part of the 2022 process. In each experiment, the training data labels are always defined by the training set applicant pools, while the test labels are defined by test applicant pools. For robustness, we provide results when flipping the training and test years in \Cref{sec:rev-appendix}. To control for the exacerbating effects of distribution shift, we show in-distribution, out-of-sample performance using our synthetic data generator in \Cref{sec:id-appendix}.

\textit{Training and evaluation.} We train a separate machine learning model for each program and admissions function using the actual training set applicant pools and outcomes.
We tailor the feature space  (removing irrelevant features and defining interaction terms) to each admission function to simplify prediction and enable near-perfect learning at train time. We train logistic regression models with $L_2$ regularization. 
We define prediction by discretizing the continuous $[0,1]$ scores. 
We rank applicants according to their score and predict 1 for the top $q$, as defined in~\eqref{eq:ml-rank}.
Empirically, we find that this improves average performance compared with using a fixed threshold as in~\eqref{eq:ml-independent}, which is not surprising in light of Proposition~\ref{prop:ml-representation}.
Results with fixed thresholds and monotonic gradient boosting are in \Cref{sec:argmax-appendix} and \Cref{sec:models-appendix}. 

\paragraph{Experiments} Our main experiments involve evaluating the performance of models trained on data from 2021 on a series of slowly shifting applicant pools.
We generate mixture pools from the 2021 and 2022 applicant pools for each school where each such pool is made of a randomly sampled $\gamma$ fraction of the applicants from 2021 and $1-\gamma$ of the applicants from 2022, for $\gamma = \{0, .1, .2, ..., 1\}$.
By construction, there are both in-distribution (in-sample) and out-of-distribution applicants in these mixtures. 
This allows us to evaluate how the addition of out-of-distribution applicants affects even in-sample performance. 
Due to stochasticity in model performance from sampling the interpolated datasets, we average performance metrics across 10 random trials. We hold admission rate constant across interpolated datasets, adjusting capacity. Results where all datasets are held at an exact fixed size are in \Cref{sec:fixed-appendix}.

\subsection{Empirical Results}

\begin{figure}[tb]
\centering
\includegraphics[width=.47\textwidth]{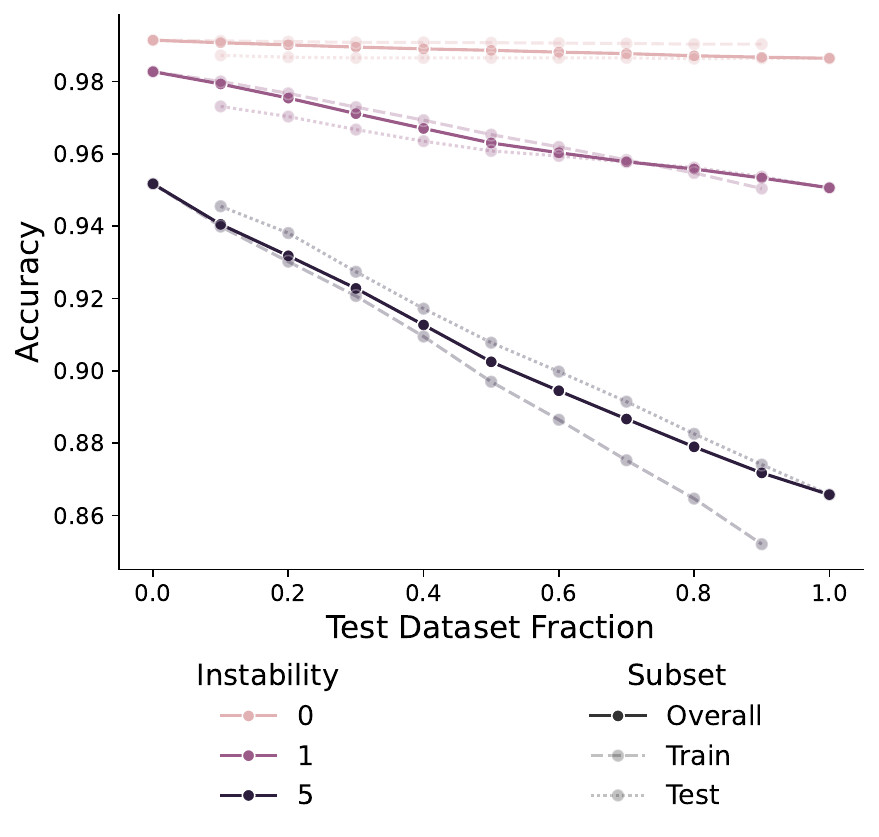}
\caption{Model accuracy as a function of choice function instability, under increasing levels of distribution shift.}
\label{fig:instability}
\end{figure}

\begin{figure}[tb]
\centering
\includegraphics[width=.47\textwidth]{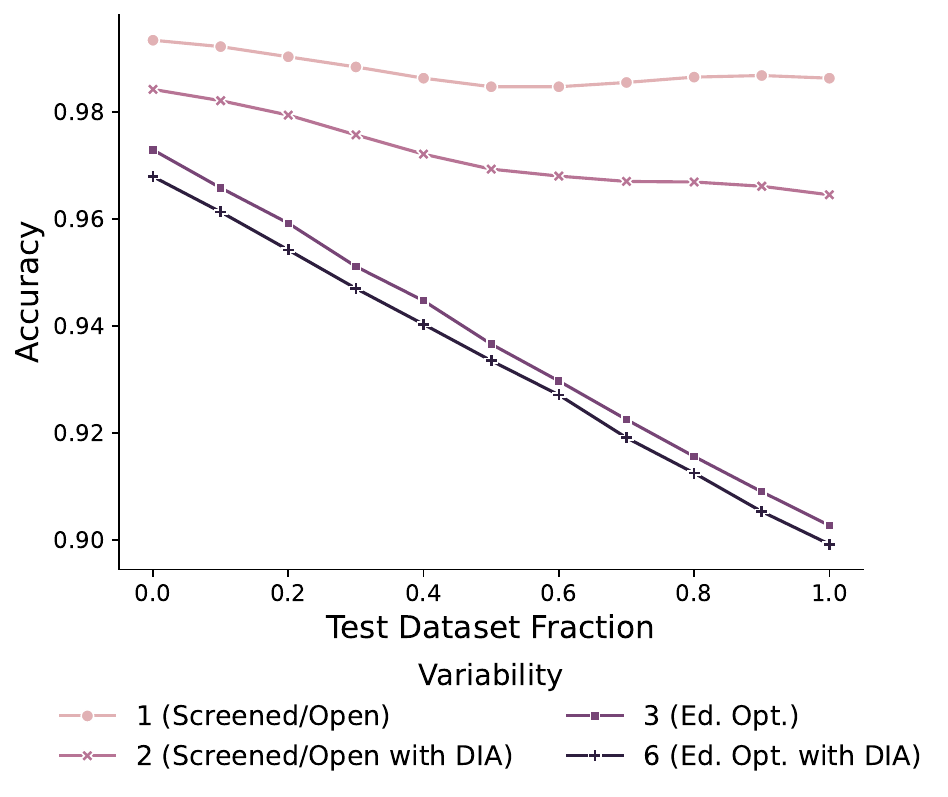}
\caption{Model accuracy as a function of choice function variability, under increasing levels of distribution shift.}
\label{fig:variability}
\end{figure}

\subsubsection{Instability}

\Cref{fig:instability} illustrates model accuracy for each $\gamma$-mixture applicant pool, averaged over all programs and choice functions with the same instability value. 
Model performance decays more steeply with out-of-distribution test data, with more stable choice functions. Furthermore, the decay does not occur just on out-of-distribution points---as would occur in standard distribution shift settings with independent outcomes.
\Cref{fig:instability} also plots the accuracy on the train and test subsets of the mixture pool separately.
In conventional machine learning models of distribution shift, performance on the in- and out-of-distribution subsets remains fixed.
Indeed, for $0$-unstable functions, performance for both subsets remains relatively constant.
In contrast, we observe that \textit{both} in- and out-of-distribution performance degrades for $d=1$ and $d=5$, with the change being more drastic for a less unstable function. At the same time, performance in the \textit{testing} subset likewise drops \textit{even as it takes up a larger portion of the overall dataset}---because the model was trained using training set application set outcomes. 

\subsubsection{Variability}

We now focus on $1$-unstable admissions functions, corresponding to actual decision processes used in NYC high school programs.
\Cref{fig:variability} shows averaged performance for admissions method grouped by variability. 
Programs with higher variability exhibit larger decays in performance.
Notably, the $1$-variable programs barely decay, in line with Proposition~\ref{prop:ml-representation}.
Larger values of variability have larger performance drops, though to a diminishing extent; the gap between variability $3$ and $6$ is small.
Our definition of variability is ``worst case'' over applicant sets, so we hypothesize that this observation is due to the fact that actual applicant sets are not adversarial, even as they exhibit distribution shift.

A more detailed breakdown of performance by admissions method can be found in \Cref{sec:per-method-appendix}.

\section{Conclusion} \label{sec:discussion}

We provide one \textit{explanatory mechanism} for why prediction admissions outcomes is hard: outcomes depend on the applicant pool, which shift over time, inducing concept drift---this dependence cannot be represented by the standard ML paradigm. As we empirically show, this dependence degrades machine learning performance, even when other challenges in this context (strategic behavior, data limitations, other forms of distribution shift) are removed.

One natural question for future work is: how do we adapt ML models to be robust to this effect? In contexts where the choice functions are algorithmic and exactly known, one can avoid training ML models and simulate the choice functions directly---in fact, this approach is the one taken for NYC High School admissions predictions, for programs in which the choice function can be simulated under alternative application pools. However, for other contexts where we cannot simulate the choice function -- such as hiring and admissions by humans -- it may be useful to develop approaches to quantify uncertainty to applicant pool changes or otherwise produce predictions robust to it, rather than naively applying conventional predictive approaches.

Finally, we note that similar concerns -- the effects of capacity constraints on predictive model performance -- may also be present in other contexts such as healthcare, e.g., hospital admissions. 

\section*{Acknowledgments}
ED is supported by National Science Foundation grant DGE-2139899. NG's work is supported by NSF CAREER IIS-2339427, NASA, the Sloan Foundation, and Cornell Tech Urban Tech Hub, Google, Meta, and Amazon research awards. SD is funded by NSF CCF 2312774, NSF OAC-2311521, NSF IIS-2442137, a gift to the LinkedIn-Cornell Bowers CIS Strategic Partnership, the AI2050 Early Career Fellowship program at Schmidt Sciences, and a PCCW Affinito-Stewart Award.

\FloatBarrier

\bibliography{bib}

\appendix

\include{appendix}

\end{document}

%% file: appendix.tex
\onecolumn

\counterwithin{theorem}{section}
\counterwithin{lemma}{section} 
\counterwithin{corollary}{section} 
\counterwithin{definition}{section} 
\counterwithin{example}{section} 

\section{Extended Theoretical Characterization and Proofs} \label{sec:proofs-appendix}
We provide proofs for supplemental results and intermediary theorems used to derive or expand on our main theorems. \Cref{sec:properties-appendix} introduces some additional properties and terms used in the literature that we utilize in our full proofs; \Cref{sec:property-lemmas} through \Cref{sec:d-unstable-appendix} elaborate on instability, culminating in \Cref{sec:instability-proof} where we explicitly connect the proof of \Cref{thm:instability} to the contents of the prior subsections. Similarly, \Cref{sec:variability-appendix} through \Cref{sec:high-variability-appendix} focus on variability, and \Cref{sec:variability-proof} proves \Cref{thm:variability} using the the previously established theorems. Lastly, we conclude this appendix section with \Cref{sec:ml-prop-proof}, where we prove \Cref{prop:ml-representation}, and \Cref{sec:program-prop-proof}, where we prove \Cref{prop:program-variability}.

\subsection{Additional, Existing Choice Function Properties}\label{sec:properties-appendix}
In addition to \Cref{def:q-acceptance} and \Cref{def:substitutability}, which we restate here, we introduce several additional properties, and provide plain English intuition for all of them:

\qacceptance*
\textit{Intuition}: Capacity of $q$.

\substitutability*
\textit{Intuition}: Students selected from a larger pool remain selected when the pool shrinks.    

\begin{restatable}[Monotonicity]{definition}{monotonicity} 
\label{def:monotonicity}
Choice function $\choice$ is \textit{monotonic} if $X_1 \subseteq X_2$ implies $\choice(X_1) \subseteq \choice(X_2)$.
\end{restatable}
\textit{Intuition}: Students selected from a smaller pool remain selected when the pool expands.    

\textbf{Note:} Monotonicity is much less studied than substitutability, but is conceptually and mathematically useful to understand as an ``inverse'' of substitutability.
    
\begin{restatable}[Consistency]{definition}{consistency} 
\label{def:consistency}
Choice function $\choice$ is \textit{consistent} if $\choice(X_2) \subseteq X_1 \subseteq X_2$ implies $\choice(X_2) = \choice(X_1)$.
\end{restatable}
\textit{Intuition}: Removing rejected candidates from the input does not change the output. 

\textbf{Note:} This is equivalent to the Independence of Irrelevant Alternatives, a common assumption in preference models, over elements. Note that this is distinct from substitutability; there exist inconsistent substitutable functions, and non-substitutable consistent functions. %

\begin{restatable}[$q$-representativeness]{definition}{qrepresentativeness} 
\label{def:q-representative}
$q$-acceptant choice function $\choice$ is $q$-representative if there exists a strict total ordering $\succ$ over all $x \in \calX$ and $\choice(X)$ is the top $q$ elements of $X$ for all $X \subseteq \calX$.
\end{restatable}

\textit{Intuition}: Students are selected according to an exact ranking. 

\textbf{Note:} We refer to this property in the main body as being ``characterized by a total order'', or a ``queue'', but this language --- and definition as a property in its own right --- is seen in, e.g., \citet{echenique2015control}.

\subsection{Property Lemmas}\label{sec:property-lemmas}

We introduce some quick lemmas about the above properties that simplify the proofs of our main results.

\paragraph{Fundamental Incompatibility}

A key structural constraint emerges immediately between monotonicity and $q$-acceptance.

\begin{restatable}[Incompatibility of Monotonicity and $q$-Acceptance]{lemma}{incompatibility}
\label{thm:mono-accept-mutex}
For $|\calX| > q \geq 1$, no choice function can be both monotonic and $q$-acceptant.
\end{restatable}

\textit{Intuition}: %
For a choice function to be monotonic, it must never reject a previously accepted applicant as the applicant pool grows. This implies that the capacity must be unbounded, which contradicts the strict capacity limit imposed by $q$-acceptance.%

\begin{proof}
Consider two different sets each of $q$ elements. If $\choice$ is $q$-acceptant, every element must be chosen in each. Then consider their union. If $\choice$ is monotonic, every element in their union would also be accepted, but that would contradict the $q$-acceptant capacity constraint. 

Consider two disjoint sets $X_1, X_2 \subseteq \calX$ with $|X_1| = |X_2| = q$. By $q$-acceptance, $\choice(X_1) = X_1$ and $\choice(X_2) = X_2$. By monotonicity, both $\choice(X_1)$ and $\choice(X_2)$ must be subsets of $\choice(X_1 \cup X_2)$. But then $|\choice(X_1 \cup X_2)| \geq |X_1| + |X_2| = 2q > q$, violating $q$-acceptance.
\end{proof}

\paragraph{Non-Substitutable Functions}
\begin{restatable}{lemma}{unsubstitutable}
\label{lem:unsubstitutable}
If $\choice$ is not substitutable, there exists some $x^*$, $x'$, $X$ and $X' = X \cup \{x'\}$ where $x^* \in X$ and $x^* \in \choice(X')$ but $x^* \not\in \choice(X)$.
\end{restatable}

\textit{Intuition}:
If a choice function is not substitutable, then there is at least one situation where a decision ``flips'' from rejection to acceptance.
We argue that there has to be some single element $x'$ that causes this ``flip'' when added.

\begin{proof}
$\choice$ is not substitutable if there exists some  $x^*$ and $X_0 \subseteq X_k$ where $x^* \in X_0$ and $x^* \in \choice(X_k)$ but $x \not\in \choice(X_0)$.
Let $k$ be such that $k = |X_k \setminus X_0|$.%

We want to show the existence of $x'$. To do this, we identify the ``closest'' sets where this occurs. For each $x \in X_k \setminus X_0$, add elements to $X_0$, constructing $X_0 \subseteq X_1 \subseteq ... \subseteq X_j \subseteq ... \subseteq X_k$, where $X_1 = X_0 \cup \{x_1\}$ and so on. Let $X_j$ be the first in this sequence of sets such that $x^* \in \choice(X_j)$ where $x^* \in \choice(X_j)$ but $x^* \not\in \choice(X_{j-1})$. Then let $x' = x_j$, $X = X_{j-1}$, and $X' = X_j$. 
\end{proof}

\subsection{Choice Distance} \label{sec:dist-appendix}

We provide some lemmas and a theorem about $r_\choice$,  defined in \Cref{def:distance}.

\begin{restatable}[Substitutability Term]{lemma}{sub-dist} 
\label{lem:substitutability-term}
Choice function $\choice$ is substitutable if and only if $|X_1 \cap \choice(X_2) \setminus \choice(X_1)| = 0$ for all $X_1 \subseteq X_2$. 
\end{restatable}

\textit{Intuition}:
For a substitutable choice function, previously rejected candidates are never accepted after the addition of more applicants. Thus labels cannot flip from negative to positive, which is what the first term counts.

\begin{proof}
For any sets $A, B$, $|A \setminus B| = 0$ if and only if $A \setminus B = \emptyset$, and furthermore $A \setminus B = \emptyset$ if and only if $A \subseteq B$. So $|X_1 \cap \choice(X_2) \setminus \choice(X_1)| = 0$ if and only if $X_1 \cap \choice(X_2) \subseteq \choice(X_1)$, which is the definition of substitutability.
\end{proof}

\begin{restatable}[Monotonicity Term]{lemma}{mono-dist} 
\label{lem:monotonicity-term}
Choice function $\choice$ is monotonic if and only if $|\choice(X_1) \setminus \choice(X_2)| = 0$ for all $X_1 \subseteq X_2$.
\end{restatable}

\textit{Intuition}:
For a monotonic choice function, previously accepted candidates are never rejected after the addition of more applicants. Thus labels cannot flip from positive to negative, which is what the second term counts.
    
\begin{proof}
Use the same set properties as per \Cref{lem:substitutability-term},  $|\choice(X_1) \setminus \choice(X_2)| = 0$ if and only if $\choice(X_1) \subseteq \choice(X_2)$, which is the definition of monotonicity.
\end{proof}

\begin{restatable}[Triangle Inequality]{theorem}{triangle} 
\label{thm:triangle-ineq}
For any choice function $\choice$ and sets $X_1 \subseteq X_2 \subseteq X_3$, $r_\choice$ obeys the triangle inequality; 
$$r_\choice(X_1, X_3) \leq r_\choice(X_1, X_2) + r_\choice(X_2, X_3)$$
\end{restatable}

Note that, in general, equality does not hold for all $X_1 \subseteq X_2 \subseteq X_3$; it is easy to construct a case where $r_\choice(X_1, X_2) + r_\choice(X_2, X_3) \neq r_\choice(X_1, X_3)$. For example:

\begin{align*}
    X_3 = \{a, b, c, d, e\}, \quad X_2 = \{a, b, c, d\}, \quad X_1 = \{a, b, c\} \\
    \choice(X_3) = \{a, e\}, \quad \choice(X_2) = \{a, d\},\quad \choice(X_1) = \{a, b\}
\end{align*}

and $r_\choice(X_1, X_2) = 1, r_\choice(X_2, X_3) = 1, r_\choice(X_1, X_3) = 1$.

\begin{proof}
See \Cref{fig:choice-sets} for notation about the partition of $X_1, X_2, X_3, \choice(X_1), \choice(X_2), \choice(X_3)$. 
\begin{figure}[h]
\centering
\begin{tikzpicture}
  \definecolor{colorone}{RGB}{204,121,167}
  \definecolor{colortwo}{RGB}{230,159,0}
  \definecolor{colorthree}{RGB}{0,114,178}
  
  \draw[colorthree, thick] (0,0.5) rectangle (10,8);
  \node[colorthree] at (0.3,5.5) {$X_3$};
  
  \draw[colortwo, thick] (1,1.75) rectangle (9.5,7.5);
  \node[colortwo] at (1.3,5) {$X_2$};
  
  \draw[colorone, thick] (2,2.25) rectangle (7,7);
  \node[colorone] at (2.3,4.5) {$X_1$};
  
  \draw[colorone, thick] (5,4.5) circle (1.75);
  \node[colorone] at (4,5.25) {$\choice(X_1)$};
  
  \draw[colortwo, thick] (7,4.5) circle (1.75);
  \node[colortwo] at (7.75,5.25) {$\choice(X_2)$};
  
  \draw[colorthree, thick] (6,2.5) circle (1.75);
  \node[colorthree] at (5.75, 1.1) {$\choice(X_3)$};
  
  \node at (4.5,4.5) {$A$};     %
  \node at (7.5,4.5) {$B$};     %
  \node at (6.5,1.5) {$C$};     %
  \node at (6,5) {$D$};     %
  \node at (5.15,3.25) {$E$};     %
  \node at (6,3.75) {$F$};     %
  \node at (6.6,3.25) {$G$};     %
  \node at (7.25,3.25) {$H$};       %
  \node at (6,2.5) {$J$};       %
  \node at (5.5,2) {$K$};     %
  \node at (6.6,5.75) {$L$};       %
  
\end{tikzpicture}
\caption{A set diagram with subsets labeled, for the proof of \Cref{thm:triangle-ineq}.} 
\label{fig:choice-sets}
\end{figure}
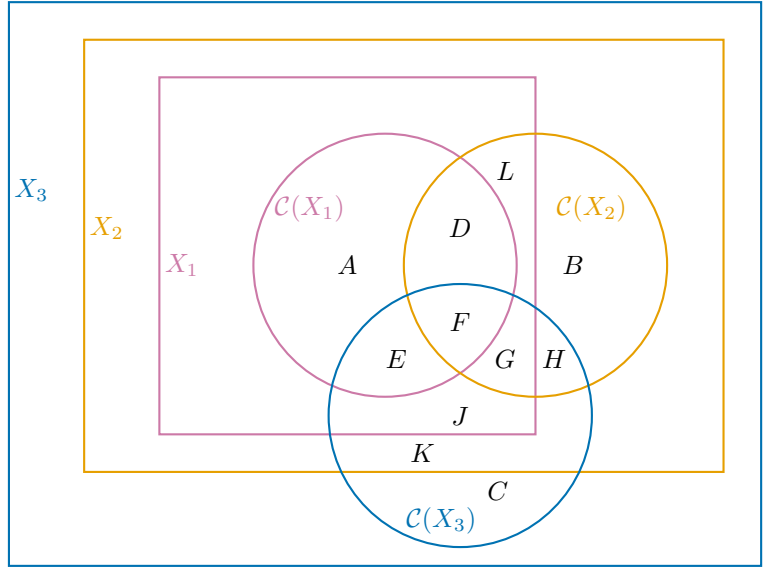
Plugging these sets into the $r_\choice$ formula, we have

\begin{align*}
    r_\choice(X_1, X_3) & = |X_1 \cap \choice(X_3) \setminus \choice(X_1)| + |\choice(X_1) \setminus \choice(X_3)| & = |J \cup G| + |A \cup D|  \\
    r_\choice(X_1, X_2) & = |X_1 \cap \choice(X_2) \setminus \choice(X_1)| + |\choice(X_1) \setminus \choice(X_2)| & = |L \cup G| + |A \cup E|  \\
    r_\choice(X_2, X_3) & = |X_2 \cap \choice(X_3) \setminus \choice(X_2)| + |\choice(X_2) \setminus \choice(X_3)| & = |E \cup J \cup K| + |L \cup D \cup B| \\
\end{align*}

Recall that set union is such that %
for any finite sets $X, Y$, $|X| \leq |X \cup Y| \leq |X| + |Y|$.
Then it becomes clear that 
\begin{align*}
    |J \cup G| + |A \cup D| & \leq |J| + |G| + |A| + |D| \\
    & \leq |E \cup J \cup K| + |L \cup G| + |A \cup E| + |L \cup D \cup B| \\
    r_\choice(X_1, X_3) & \leq r_\choice(X_1, X_2) + r_\choice(X_2, X_3)
\end{align*}
\end{proof}

\subsection{0-Instability}\label{sec:0-unstable-appendix}

\begin{restatable}[$0$-Instability]{definition}{zerostability} 
\label{def:zero-instability}
Choice function $\choice$ is 0-unstable if and only if $r_\choice(X_1, X_2) = 0$ for all $X_1 \subseteq X_2$.
\end{restatable}

This describes realizable machine learning settings, where labels do not change as a result of shifts in the applicant pool.

\begin{restatable}{corollary}{Zero-Distance} 
\label{thm:sub-mon-zero}
Choice function $\choice$ is 0-unstable if and only if it is both substitutable and monotonic.
\end{restatable}

This immediately follows from the two terms comprising $r_\choice$ mapping onto substitutability and monotonicity.

\subsubsection{Independence}

\begin{restatable}{definition}{Independence} 
\label{def:independence}
$\choice$ is independent if for 
 for all $x$,
either $x \in \choice(X)$ 
or $x \not\in \choice(X)$ for all $X$ where $x \in X$. 
\end{restatable}

In other words, any element is either always accepted, or always rejected, when available, irrespective of the other elements available in a set $S$.

\subsubsection{Equivalence}

\begin{restatable}{theorem}{Independent Rules are Substitutable and Monotonic} 
\label{thm:independent-sub-mon}
$\choice$ is independent if and only if it is both substitutable and monotonic.
\end{restatable}

\textit{Intuition}: The only way for adding applicants to neither hurt those already accepted or help those already rejected is to not affect them at all.
\begin{proof}
    Assume $\choice$ is independent.

    \textbf{Informal}
    Consider how independence implies substitutability and monotonicity. Independence means any element is always accepted or always rejected regardless of applicant pool. Then any $x$ that is accepted in $X_1$ is still accepted in any subset of $X_1$, which satisfies the definition of substitutability, and is accepted in any superset of $X_1$, which satisfies monotonicity.

    To prove that substitutability and monotonicity combine to form independence, consider the behavior they cover. Recall the intuition behind substitutability --- removing applicants does not lead to rejecting those accepted, and the intuition behind monotonicity --- adding applicants does not lead to rejecting those accepted. Given both properties, nothing can affect an acceptance decision at all. 

    \textbf{Formal}
    First, we assume independence and prove $0$-stablility.
    Recall \Cref{lem:substitutability-term}: $\choice$ is substitutable if and only if $|X_1 \cap \choice(X_2) \setminus \choice(X_1)| = 0$ for all $X_1 \subseteq X_2$. Then for any $x \in X_1$, either $x \in \choice(X_1),\choice(X_2)$ or $x \not\in \choice(X_1),\choice(X_2)$, so $X_1 \cap \choice(X_2) = \choice(X_1)$. Then $|X_1 \cap \choice(X_2) \setminus \choice(X_1)| = 0$ and so $\choice$ is substitutable. Similarly, recall \Cref{lem:monotonicity-term}: $\choice$ is monotonic if and only if $|\choice(X_1) \setminus \choice(X_2)| = 0$ for all $X_1 \subseteq X_2$. Any element $x \in \choice(X_1)$ is in $X_2$ because $\choice(X_1) \subseteq X_1 \subseteq X_2$. By definition of independence, $x \in \choice(X)$ for all $X$ where $x \in X$, so if $x \in \choice(X_1)$, $x \in \choice(X_2)$, so $\choice(X_1) \subseteq \choice(X_2)$.

    To prove that substitutability and monotonicity together imply independence, consider any $x \in X$. If $x \in \choice(X)$, we can show that $x \in \choice(X')$ for any $X'$ where $x \in X'$. By monotonicity, $x \in \choice(X' \cup X)$, and so by substitutability, $x \in \choice(X')$. If $x \not\in \choice(X)$, we can prove $x \not\in \choice(X')$ for any $X'$. If we assumed for contradiction that $x \in \choice(X')$, then the same argument as above proves $x \in \choice(X)$, leading to contradiction.
\end{proof}

\begin{restatable}{corollary}{Only Independent Rules are 0-unstable} 
\label{thm:independent-zero-unstable}
$\choice$ is $0$-unstable if and only if it is independent, and therefore no $q$-acceptant $\choice$ can be $0$-unstable.
\end{restatable}

This immediately follows from $\Cref{thm:sub-mon-zero}$, $\Cref{thm:independent-sub-mon}$, and $\Cref{thm:mono-accept-mutex}$.

\subsection{1-Instability} \label{sec:1-unstable-appendix}

\subsubsection{Equivalence with Substitutability} \label{sec:sub-equiv-proof}

\begin{restatable}[Substitutability and 1-Instability are Equivalent Under Capacity Constraints]{theorem}{subequiv}
\label{thm:sub-unstable-equiv}
Any $q$-acceptant choice function is $1$-unstable if and only if it is substitutable.
\end{restatable}

\begin{proof}
\textbf{Informal}
Intuitively, $q$-acceptance means that every acceptance decision that becomes a rejection must also be paired with a new acceptance to keep the size of the accepted set the same. Substitutability implies that adding elements cannot cause new acceptances from rejections, so the only possible newly accepted element is the newly added one.  

Recall that the first term of $r_\choice$ is always zero only for substitutable $\choice$. So, we only need to show $|\choice(X) \setminus \choice(X')| \leq 1$ if and only if $\choice$ is substitutable. %
First, we show that $\choice$ being substitutable is \textit{sufficient} for $1$-instability. If $\choice$ is substitutable, then all elements in $\choice(X')$ except for $x'$ must be in $\choice(X)$. By $q$-acceptance, $|\choice(X')| = |\choice(X)| = q$ (the $|X| < q$ case is trivial), so $\choice(X)$ can have at most one element not in $\choice(X')$, making $\choice$ 1-unstable. 

Second, we show that substitutability is \textit{necessary} for $1$-instability. If $\choice$ is not substitutable, there is some $X \subseteq X'$ and $x^* \in X$ such that $x^* \notin \choice(X)$ but $x^* \in \choice(X')$ (by \Cref{lem:unsubstitutable}). This creates one change. By $q$-acceptance, this forces another candidate out of $\choice(X')$, creating at least a second change, violating 1-instability.

Together, these show that under $q$-acceptance, substitutability is necessary and sufficient for $1$-instability, completing the proof.
\end{proof}

\begin{restatable}{theorem}{BoundedDistance} 
\label{thm:sub-accept-consistent}
Any $q$-acceptant, substitutable $\choice$ is consistent.
\end{restatable}

\begin{proof}
Consider $X_1, X_2$ where $\choice(X_2) \subseteq X_1 \subseteq X_2$ for substitutable, $q-$acceptant $\choice$. Substitutability guarantees that if $x \in \choice(X_2)$ and $x \in X_1$, $x \in \choice(X_1)$. By construction, this applies to all $x \in \choice(X_2)$, because $\choice(X_2) \subseteq X_1$; therefore, $\choice(X_2) \subseteq \choice(X_1)$. As $|\choice(X_2)| = |\choice(X_1)|$, and $\choice(X_1),\choice(X_2)$ are finite sets, this means that $\choice(X_2) = \choice(X_1)$, which is the definition of consistency.
\end{proof}

\subsection{\textit{d}-Instability} \label{sec:d-unstable-appendix}

\subsubsection{Proof of \Cref{lem:unstable-calc}} \label{sec:calc-proof}

\begin{restatable}[Calculating Instability]{lemma}{stablecalc} 
\label{lem:unstable-calc}
For all $q$-acceptant $\choice$, let $X' = X \cup \{x'\}$ for $|X| \geq q$ and $n = |\choice(X) \setminus \choice(X')|$.  Then 
$r_\choice(X, X') = 2n-\mathbf{1}_{x'\in\choice(X')}$. 
\end{restatable}

\begin{proof}
\textbf{Informal}
If input sets differ by one element, but the selected sets differ by $n$ elements, then it must be that $n$ previously accepted elements were rejected and there must be $n$ new accepted elements, at least $n-1$ must have been previously rejected, because the size of the accepted set is fixed. 

\textbf{Formal}
By $q$-acceptance, $|\choice(X)| = |\choice(X')| = q$. Then $|\choice(X) \setminus \choice(X')| = |\choice(X') \setminus \choice(X)| = q-|\choice(X) \cap \choice(X')|$. 
Second, note that $\choice(X') \setminus \choice(X) = (\choice(X') \cap X) \setminus \choice(X)$ if $x' \not\in \choice(X')$, and otherwise $\choice(X') \setminus \choice(X) = \{x'\} \cup \left((\choice(X') \cap X) \setminus \choice(X)\right)$. 
Putting the previous statements together, we get two cases:
\begin{itemize}
    \item $x' \not\in \choice(X')$: $|\choice(X) \setminus \choice(X')| = |\choice(X') \setminus \choice(X)| = |X \cap \choice(X') \setminus \choice(X)|$ 
    \item $x' \in \choice(X')$: $|\choice(X) \setminus \choice(X')| = |\choice(X') \setminus \choice(X)| = |X \cap \choice(X') \setminus \choice(X)| + 1$ 
\end{itemize}
Then $r_\choice(X, X') = |X \cap \choice(X') \setminus \choice(X)| + |\choice(X) \setminus \choice(X')| = 2|X \cap \choice(X') \setminus \choice(X)| + \mathbf{1}_{x'\in\choice(X')} = 2|\choice(X) \setminus \choice(X')| - \mathbf{1}_{x'\in\choice(X')}$
\end{proof}

\begin{example}
Suppose applicants are each one of at least 
$n$ different types (e.g. instrument specializations) and suppose they also have an overall ``talent'' rating.
Then suppose that admissions is determined by forming complementary groups of size $n$ (e.g. quartets) to the extent possible, and is otherwise based on a single ``talent'' rating.
Suppose type 1 candidates are the bottleneck in forming groups. 
Then one additional type 1 candidate in the applicant pool allows one additional group to be admitted ($n$ new admits). By \Cref{lem:unstable-calc}, this rule cannot be $d$ unstable for $d<2n-1$.
\end{example}

\subsubsection{Inconsistent Instability}\label{sec:inconsistent-proof}

\begin{restatable}[Even Instability is Inconsistent]{theorem}{inconsistentstable} 
\label{thm:even-inconsistency}
For any $q$-acceptant $\choice$, $\choice$ is inconsistent if and only if there exist $X, X' = X \cup \{x'\}$ where $r_\choice(X, X')$ is positive and even. 
\end{restatable}

\textit{Intuition}: for $\choice$ to have an even but nonzero choice distance $r_\choice$, the indicator term in  \Cref{lem:unstable-calc} must be zero; the added $x'$ is not selected. But that means $x'$ is a rejected element that still has an impact, which violates consistency.
\begin{proof}
\textbf{Informal}
First, we prove that an even and positive value of $r_\choice$ implies inconsistency. For $\choice$ to have an even but nonzero choice distance $r_\choice$, the indicator term in  \Cref{lem:unstable-calc} must be zero; the added $x'$ is not selected. But that means $x'$ affects selection without actually being chosen, which violates consistency.

Next, prove that consistency limits the values of $r_\choice$. If $\choice$ is not consistent, then there is some $X, X'$ where the new element $x'$ is not chosen (so the indicator term in \Cref{lem:unstable-calc} is zero and $r_\choice$ is even), but changes the chosen elements. This makes  $r_\choice$ is positive under a fixed capacity constraint.

\textbf{Formal}
Consider \Cref{lem:unstable-calc}. See that for $r_\choice \neq 0$, $r_\choice$ is odd if and only if $x' \in \choice(X')$ and even if and only if $x' \not\in \choice(X')$. If $x' \not\in \choice(X')$, then $\choice(X') \subseteq X$. %
So if $x' \not\in \choice(X')$ but $n > 0$, $\choice(X') \subseteq X$ but $\choice(X') \neq \choice(X)$, which proves inconsistency.
Similarly, if $\choice$ is inconsistent, there is $\choice(X') \subseteq X$ but $\choice(X') \neq \choice(X)$. %
As $\choice(X') \neq \choice(X)$, $n > 0$, but for $\choice(X') \subseteq X$ to be true, $x' \neq \choice(X')$, and so $r_\choice(X, X')$ is even. 
\end{proof}

\begin{restatable}{corollary}{No Consistent Tightly Even Instability} 
\label{cor:even-inconsistency}
There is no $q$-acceptant and consistent $\choice$ that is tightly $dk$-unstable for even values of $d$.
\end{restatable}

This follows directly from \Cref{thm:even-inconsistency}.

While inconsistent choice functions are generally less intuitive (and formally studied much less), there exist some real-world examples of inconsistent choice functions. Famously, voting theory is rife with spoiler effects, where adding a dummy candidate changes the selected outcome. Consider the choice function in the US presidential election of 1912. Say that with the candidate pool being just William Taft and Woodrow Wilson, Taft would win (Wilson only got 41\% of votes). However, adding Theodore Roosevelt to the ballot led to Wilson winning, which violates consistency. 
More relevantly to our setting, some schools value student uniqueness. We could imagine cases where a university admits a student because their niche interest is particularly noteworthy, but if multiple have this interest, they may end up taking none. Similarly, if only one student from an unknown high school applies to a university, they may stand out more than if every student from that school applied. 

\subsection{Proof of \Cref{thm:instability}} \label{sec:instability-proof}

Recall \Cref{thm:instability}:
\thminstability*

Putting things together, the first statement of \Cref{thm:instability} comes from \Cref{thm:independent-zero-unstable} and the second from \Cref{thm:sub-unstable-equiv}. The third 
has roots in \Cref{lem:unstable-calc}, which allows us to calculate instability.

Rigorously, we can generate a $d$-unstable $\choice$ for any $1 < d \leq 2q$. Take any $1$-unstable, $q$-acceptant $\choice$ induced by a total order $\succ$ (i.e, is a ranking). Let $n$ be $\frac{d}{2}$, rounded down. Then let $X^* = \{x^*_1, ..., x^*_n\}$ be the $n$ \textit{minimal} elements of $\succ$ (or any $n$ elements that are not among the $n$ maximal elements, if there exist any.) Then modify $\choice$ such that if $X^* \subseteq X$ for any input $X$, $X^* \subseteq \choice(X)$. Then $\choice$ is tightly $2n-1$-unstable. To make $\choice$ $2n$-unstable, let there be some $x'$ such that if $X^* \subseteq X$ but $x' \in X$ as well, then \textit{no} element of $X^*$ is in $\choice(X)$. %

The intuition for the construction is that $X^*$ forms a highly desirable team, even though each element is not individually desirable, nor are partial subsets of $X^*$.
In the inconsistent case (even instability) $x'$ represents a candidate whose mere presence negates the value of $X^*$ as a team, even when not accepted. 

\subsection{Variability} \label{sec:variability-appendix}

\subsubsection{Generalizing and Justifying Variability}

In \Cref{sec:theory-definitions}, we provide a definition of variability for $1$-unstable, $q$-acceptant functions. Here, we provide a more thorough justification of \Cref{def:variability} by reducing it from a more general and nuanced characterization.
Notably, instability counts how many admissions decisions change when an applicant is added \textbf{or} removed. The definition in \Cref{def:variability}, however, focuses only on the case of \textbf{adding} elements. In the general case, there are two relevant sets, defined when adding or removing an element. 

As per \Cref{def:variability}, when appending a new element to the input set $X$, different applicants may become rejected: let the set of ``borderline admits'' $V_\choice(X)$ be 
\begin{equation}
V_\choice(X) = \bigcup\limits_{x'\in\calX} \choice(X) \setminus \choice(X \cup \{x'\})
\end{equation}
which we shall call the ``borderline set'' for short. Notice that this is bounded by $q$, as $|\choice(X)| = q$.

Conversely, when removing an applicant from $X$, let the set of applicants that might be newly accepted be
\begin{equation}
A_\choice(X) = \bigcup_{x^*\in X} \choice(X \setminus \{x^*\}) \setminus \choice(X)
\end{equation}
the set of ``waitlisted rejections''; likewise, we shall call this the ``waitlisted set''. 

\begin{definition}[Variability, General]
\label{def:variability-general}
A choice function $\choice$ has \textit{variability} $m$ where:
$$m := \max\left\{\max_{X\subset \calX} \left|V_\choice(X)\right|, \max_{X\subset \calX} \left|A_\choice(X)\right| \right\}$$
\end{definition}

\textit{Intuition}: Variability bounds how many different currently accepted candidates could potentially be displaced by adding any single new candidate \textit{or} how many different candidates are newly accepted if a currently accepted candidate is removed.%

\subsubsection{Equivalence in the 1-unstable Case} \label{sec:var-proof}

Here, we show that, for capacity-constrained, substitutable functions in large universes, the general definition in \Cref{def:variability-general} reduces to \Cref{def:variability}.

\begin{restatable}{theorem}{appendremove} 
\label{thm:append-remove-size}
For $q$-acceptant, 1-unstable $\choice$, and $|\calX| \geq 2q$, the maximum size of the waitlisted set is equal to the maximum size of the borderline set; that is, 
\begin{equation}
\max_{X\subset \calX} \left| V_\choice(X)\right| = \max_{X\subset \calX} \left| A_\choice(X)\right|
\end{equation}
\end{restatable}

In short, if there is a set $X$ with a waitlisted set of size $u$, there must be a set $X_0$ with a borderline set of size (at least) $u$, and vice versa. %

\begin{proof}
To prove equality, we will prove that 1) $\max_{X\subset \calX} \left| V_\choice(X)\right| \geq \max_{X\subset \calX} \left| A_\choice(X)\right|$ and then that 2) $\max_{X\subset \calX} \left| V_\choice(X)\right| \leq  \max_{X\subset \calX} \left| A_\choice(X)\right|$.

\textbf{Part 1).}
To prove the first statement, assume that there is some $X$ where the waitlisted set is of cardinality at least $u$, so $x^*_1, ..., x^*_u \in A_\choice(X)$.

Then (recall that $\choice$ is consistent) there are $X_1, ..., X_u \subseteq X$ where $X_i = X \setminus \{x'_i\}$ such that $\choice(X_i) = x^*_i \cup \choice(X) \setminus x'_i$ for $i \in \{1, ..., u\}$. Then consider $X_0 = \bigcap_{i=1}^u X_i = X \setminus \{x'_1, ..., x'_u\}$, and $\choice(X_0)$. It must be that $\choice(X_0) = \{x^*_1, ..., x^*_u\} \cup \choice(X) \setminus \{x'_1, ..., x'_u\}$ by substitutability. I claim that $\{x^*_1, ..., x^*_u\} \subseteq V_\choice(X_0)$, which means $|V_c(X_0)| \geq u$, which would complete the proof.

Assume for contradiction that this is untrue; $\{x^*_1, ..., x^*_u\} \not\subseteq V_\choice(X_0)$. Then some $x^*_i$, say $x^*_u$ (without loss of generality), is not in $V_\choice(X_0)$, which means that $x^*_u \in \choice(\{x'_u\} \cup X_0)$.   
Consider the makeup of $\choice(\{x'_u\} \cup X_0)$. Recall that $\{x'_u\} \cup X_0 = \bigcap_{i=1}^{u-1} X_i = X \setminus \{x'_1, ..., x'_{u-1}\}$. Because of substitutability, as a subset of $X$, we have $\left(\choice(X) \setminus \{x'_1, ..., x'_{u-1}\}\right) \subseteq \{x'_u\} \cup X_0$. From substitutability, as a subset of from $X_1$ through $X_{u-1}$, $\{x^*_2, ..., x^*_{u-1}\}$ are likewise chosen. 
However, consider the size of $\choice(\{x'_u\} \cup X_0)$. $|\choice(X) \setminus \{x'_1, ..., x'_{u-1}\}|$ is of size $q-(u-1)$, because $x'_i \in \choice(X)$ for all $i$. We also know that $|\{x^*_1, ..., x^*_{u-1}\}| = u-1$, and as no  $x^*_i$ is in $\choice(X)$ by definition, all of these elements are distinct. As $x^*_u$ remains by assumption, this adds up to $|\choice(\{x'_u\} \cup X_0)| = q+1$, which violates $q$-acceptance and completes the proof.

\textbf{Part 2).}
The proof for 2) is extremely similar to the former. Begin with $X_0$ with $V_\choice(X_0) = \{x^*_1, ..., x^*_m\}$, with $x'_1, ..., x'_m$ such that $\choice(X'_1) = \{x'_1\} \cup \choice(X_0) \setminus \{x^*_1\}$. Let $X = \bigcup_{i=1}^m /X'_i = X_0 \cup \{x'_1,...,x'_m\}$ and $X_i = X \setminus \{x'_i\}$. Note that $X_i = X'_1 \cup ... X'_{i-1} \cup X'_{i+1} \cup ... \cup X_m$. %
Because of $1$-instability, $|\choice(X_0) \setminus \choice(X)| \leq |X_0 \setminus X| = m$. Because of substitutability, any $x \in \choice(X)$ is either in $\choice(X_0)$ or not in $X_0$ and therefore in $\{x'_1, ..., x'_m\}$. 

I claim $\choice(X) = \{x'_1,...,x'_m\} \cup \choice(X_0) \setminus \{x^*_1,..., x^*_m\}$. Firstly, $x^*_i \not\in \choice(X)$ because  $x^*_i \not\in \choice(X'_i)$ and $X'_i \subseteq X$, so that would be a substitutability violation. As $\choice(X_0) \setminus \{x^*_1,..., x^*_m\}$ has $q-m$ elements, the only possible elements to fill $\choice(X)$ must be the $m$ elements of $\{x'_1,...,x'_m\}$. 

The same argument holds for any $X_i$, but use $X_m$ without loss of generality for ease of notation; $X'_j \subseteq X_m$ for $m \neq j$, so $x^*_j \not\in\choice(X_m)$ by substitutability. With $\choice(X_0) \setminus \{x^*_1,..., x^*_{m-1}\}$ being the only $q-m+1$ elements of $X_0$ that can be in $\choice(X_m)$, the other $m-1$ elements must be $\{x'_1,. .., x'_{m-1}\}$. Then consider the definition of $A_\choice(X)$ and see that $\choice(X \setminus x'_i) \setminus \choice(X) = \{x^*_i\}$ for any $i$, with each $x^*_i$ being unique.
\end{proof}

\subsection{Queues and 1-Variability}\label{sec:ranking-proof}

\subsubsection{Some Lemmas}

\begin{restatable}[Consistency of Removable Sets]{lemma}{consistentremove}
\label{lem:consistent-remove}
Let $\choice$ be 1-unstable and q-acceptant with variability $m$. For any $X_1, X_2$, if $\choice(X_1) = \choice(X_2)$, then $V_\choice(X_1) = V_\choice(X_1)$.
\end{restatable}

\begin{proof}
Recall from \Cref{thm:sub-accept-consistent} that $\choice$ must be consistent and so $\choice(\choice(X)) = \choice(X)$ for any $X$. 

Consider any $x'$ associated with an element $x^*$ of $V_\choice(X)$, such that $\choice(X) \setminus \choice(X') = \{x^*\}$ (and even more precisely, by $q$-acceptant substitutability, that $\choice(X') = \{x'\} \cup \choice(X) \setminus \{x^*\}$ exactly). Note that $\choice(X') \subseteq \choice(X) \cup \{x'\} \subseteq X'$, and therefore $\choice(\choice(X) \cup \{x'\}) = \choice(X')$ by consistency. This applies to all possible $x'$ (including when $x' \in X$), and therefore, $V_\choice(X) = V_\choice(\choice(X))$.

As this set equality is exact, apply this logic to some $X_2$ and see $V_\choice(X_1) = V_\choice(\choice(X_1)) = V_\choice(\choice(X_2)) = V_\choice(X_2)$.
\end{proof}

\begin{restatable}{lemma}{Ranking m} 
\label{lem:a-v-equal}
Let $q$-acceptant, $1$-unstable $\choice$ have variability $1$. Then $V_\choice(X) = A_\choice(X \cup \{x'\})$ when $\choice(X \cup \{x'\}) \neq \choice(X)$.
\end{restatable}

\begin{proof}
If $\choice(X \cup \{x'\}) \neq \choice(X)$, then $\choice(X \cup \{x'\}) = \{x'\} \cup \choice(X) \setminus x^*$ where $V_\choice(X) = \{x^*\}$ (by consistency, 1-instability, and the definition of variability). Then $X' \setminus \{x'\}) = X$ and so $A_\choice(X') = \{x^*\}$, completing the proof. 
\end{proof}

\subsubsection{Equivalence of Queues and $1$-variability}
Here we present the proof for $1$-unstable $1$-variability being equivalent to $q$-representativeness, i.e., being fully characterized by a total order.

\begin{restatable}[$q$-Representativeness and $1$-Variability are Equivalent]{theorem}{rankingvariability} 
\label{thm:ranking-variability}
$\choice$ is $q$-representative if and only if it is $q$-acceptant and $1$-unstable with variability $1$.
\end{restatable}
\begin{proof}
We prove this in two parts: first, that $q$-representativeness implies $1$-variability, and then the converse.

\paragraph{Direction 1} %
Note that $q$-representativeness implies $q$-acceptance by definition. 

Next, we prove that $q$-representativeness implies $1$-instability. Assume for contradiction that some $q$-representative $\choice$ is not substitutable (and therefore not $1$-unstable). Then there exists (by \Cref{lem:unsubstitutable}) some $x^*$, $x'$, $X$ and $X' = X \cup \{x'\}$ where $x^* \in X$ and $x^* \in \choice(X')$ but $x^* \not\in \choice(X)$. By $q$-acceptance, there must be at least one $x$ in $\choice(X)$ that is not in $\choice(X')$ (which must take the place of $x^*$). However, $\choice(X')$ implies that $x^* \succ x$, while $\choice(X)$ implies that $x \succ x^*$ --- violating asymmetry and leading us to contradiction.

Lastly, proving that $q$-representative $\choice$ leads to $m=1$ is relatively intuitive. For any $X$, $\choice(X)$ are the $q$ highest-ranked elements of $X$. Denote $X = \{x_1, x_2, ..., x_k\}$ where $|X| = k$ and $x_1 \succ x_2 \succ ...\succ x_k$. Then $\choice(X) = \{x_1, ..., x_{q^*}\}$ where $q^* = \min\{q, k\}$. 
Assume for contradiction that there is a $q$-representative $\choice$ with variability $m > 1$. Then there exists some $X$ where $|V_\choice(X)| > 1$; say that $\{x^*_a, x^*_b\} \subseteq V_\choice(X)$. Then there exists some $x'_a$ where $\choice(X) \setminus \choice(X \cup x'_a) = \{x^*_a\}$ and similar $x'_b$ for $x^*_b$. The first implies $x^*_b$ is ranked above $x^*_a$, but the second implies the opposite. Then choices cannot be made according to a static ranking and $\choice$ is not $q$-representative, completing the proof by contradiction. %

\paragraph{Direction 2}
Assume $q$-acceptance, variability 1, and instability 1. Now we prove $\choice$ is $q$-representative. To be precise, I take this to mean that $\choice$ is indistinguishable from some $q$-representative $\choice$; $\choice$ could be represented by some total ordering over elements of $X$. We show this strict total ordering constructively.

Define $\succ$ such that, if $x_a \in \choice(X)$ but $x_b \not\in\choice(X)$ for any $x_a, x_b \in X$, $x_a \succ x_b$. Note that this implies that the existence of some set $A$ where $x_a \in \choice(A)$ and $A_\choice(A) = \{x_b\}$ by consistency (e.g. $A = \choice(X) \cup \{x_b\}$). This also implies $A'$ where $x_a, x_b \in \choice(A')$, $V_\choice(A') = \{x_b\}$, where $A' = A \setminus \{x\}$ for any $x \neq x_a, x_b$. 

Then, we need to prove this relation is asymmetric, transitive, and completeness (or rather, is indistinguishable from a complete relation when examining the induced choice function).

\paragraph{Asymmetry} 
First, we show asymmetry. Assume for contradiction that $a \succ b$ and $b \succ a$. Then there exists $X_a$ where $a, b \in X_a$, $a \in \choice(X_a)$ and $b \not\in \choice(X_a)$, and a similar $X_b$ where $a \not\in \choice(X_b)$ and so on.

Let $X_a = \choice(X_a) \cup \{b\}$ (we can construct such out of any $X_a$, by consistency), where $A_\choice(X_a) = \{b\}$, and denote $\choice(X_a) = \{x^a_1, ..., x^a_{q-1}, a\}$. Construct $X_b$ similarly. Further, let $X'_a = X_a \setminus \{x^a_1\}$, and note that $V_\choice(X_a') = \{b\}$.

Then from here, we add the elements of $X_b$ to $X_a'$, constructing $X_1, X_2,$ and so forth. 

Let $X_0 = X_a'$ and sequentially construct a series of $X_t$. Add $x^b_t$ to $X_{t-1}$. If $x_b \not\in \choice(X_{t-1} \cup \{x^b_t\})$, then let $X_t = X_{t-1} \cup \{x^b_t\}$, and $V_\choice(X_t) = \{b\}$ by \Cref{lem:consistent-remove}. If $x^b_t \in \choice(X_{t-1} \cup \{x^b_t\})$, then $b \not\in \choice(X_{t-1} \cup \{x^b_t\})$ and $A_\choice(X_{t-1} \cup \{x^b_t\}) = \{b\}$ (by \Cref{lem:a-v-equal}). Then let $X_t = \{x^b_t\} \cup X_a' \setminus x^a_i$ where $x^a_{i-1} \not\in X_{t-1}$ but $x^a_i \in X_{t-1}$, and $V_\choice(X_t) = \{b\}$ (again, by \Cref{lem:a-v-equal}). Note that $a \in \choice(X_t)$, $b \in X_t$, and specifically $b \in V_\choice(X_t)$, for all $t \in \{0,..., q-2\}$.

Then $X_{q-2}$ will contain $x^b_1,...x^b_{q-2}$. Let $X^* = X_{q-2} \cup \{x^b_{q-1}\}$. Note that adding $x^b_{q-1}$ to $X_{q-2}$ makes $X^* \supseteq X_b$. Then, if $a \in \choice(X^*)$, substitutability is violated, as $a \not\in \choice(X_b)$. If $x^b_{q-1} \not\in \choice(X^*)$, then $\choice(X_{q-2}) = \choice(X^*)$ by consistency, and $a \in \choice(X_{q-2})$. If instead $x^b_{q-1} \in \choice(X^*)$, recall that $V_\choice(X_{q-2}) = \{b\}$, so $a \not\in V_\choice(X_{q-2})$, and $a$ remains in $\choice(X^*)$. Therefore, either possibility leads to contradiction, proving that $\succ$ is asymmetric.

\paragraph{Transitivity}
Say that $a \succ b$ and $b \succ c$ and we want to prove $a \succ c$. If $a \succ b$, there exists some $X$ where $a, b \in X$, $a \in \choice(X)$, and $b \not\in \choice(X)$. Let $X$ be $\choice(X) \cup \{b\}$, where $a \in \choice(X), b \not\in\choice(X)$ by consistency. Then consider $X \cup \{c\}$. By asymmetry, $c \not\succ b$, so $c \not\in \choice(X\cup \{c\})$. By consistency, then, $\choice(X\cup \{c\}) = \choice(X)$, so $a \in \choice(X\cup \{c\})$ and $a \succ c$, completing the proof.

\paragraph{(Almost)-Completeness}
See that this relation must exist for all but $q$ elements --- and any arbitrary linear extension thereof generates the same choice function behavior.

For any $a, b$, let there be some $X_0$ where $a, b \not\in \choice(X_0)$. Then remove elements of $\choice(X_t)$ in a sequence $X_0, ..., X_t$ until $A_\choice(X_t) = \{a\}$ or $\{b\}$. This must happen for some value of $t$ because when $|X_t| = q+1$, at least one of $a$ or $b$ are in $\choice(X)$, and due to $1$-instability, both $a$ and $b$ cannot be added to $\choice(X_t)$ at once. The only exception is if there is no such $X_0$ where $a, b \not\in \choice(X_0)$, which can only be true if $a$ is such that $a \in \choice(X)$ for all $X$ (and similarly for $b$). (If $a \not\in \choice(X)$ for some $X$, then substitutability implies $a \not\in \choice(X \cup \{b\})$). However, there can only be $q$ such elements (by $q$-acceptance). Note that any linear extension of this partial order will produce the same choice function behavior; that is, any constructed total ordering that otherwise aligns with $\succ$ will produce the same choice behavior, and thus, $\choice$ can be represented by any such strict total ordering.
\end{proof}

\subsection{Higher-Variability Functions} \label{sec:high-variability-appendix}
In this subsection, we examine two different cases of choice functions: those that are sequentially composed of other choice functions, and those that can be represented as a linear assignment problem (also described in graph-theoretic terms as a maximum- or minimum-weight bipartite matching).

\subsubsection{Sequential Composition} \label{sec:seq-proofs}
A useful class of choice functions are choice functions that can be broken down or defined as a sequential composition of functions $\choice_n, \choice_{n-1}, ..., \choice_1$; that is, 
$\choice(X) = \choice_n(X) \cup \choice_{n-1}(X_{n-1}) \cup ... \cup \choice_1(X_1)$ where $X_i = X_{i+1} \setminus\choice_{i+1}(X_{i+1})$ and $X_n = X$. Sequential composition is particularly useful because it preserves a variety of useful properties. Sequences of $q_i$-acceptant functions are $q$-acceptant, because no duplicate elements are selected by multiple functions, and therefore preserves $q$-acceptance when each function in the sequence has capacity $q_i$ (i.e., $q = q_n + ... + q_1$.) 
\seqcomposition*

\textit{Intuition}: Selection proceeds in stages, with each stage operating on candidates not selected by higher-priority stages.

\paragraph{Sequentially Composed Substitutable Functions Are Substitutable}
Similarly, we show sequential compositions of substitutable functions are also substitutable.
\begin{restatable}[Composition Preserves Substitutability]{theorem}{seqsub}
\label{thm:composition-substitutable}
Sequential composition of substitutable functions yields a substitutable function.
\end{restatable}

\textit{Intuition}:
``Appending'' a substitutable choice function $\choice_2$ before another substitutable function $\choice_1$ affects the inputs to $\choice_1$ in a structured way.
\begin{proof}
    Proof by induction.
    In the base case, for $n = 1$, $\choice(X) = \choice_1(X)$, and as $\choice_1$ is given to be substitutable, $\choice$ is as well.

    For the inductive step, show that $\choice^{n+1}$ is substitutable when comprised of $\choice_{n+1}, \choice_n, ..., \choice_1$, assuming that $\choice^n$ comprised of $\choice_n,...,\choice_1$ is substitutable. 

    The definition of substitutability is that $X \subseteq X'$ implies $\choice(X') \cap X \subseteq \choice(X)$. So we need to show $\choice^{n+1}(X') \cap X \subseteq \choice^{n+1}(X)$. 
    
    First, note that $\choice^{n+1}(X) = \choice_{n+1}(X) \cup \choice^n(X \setminus \choice_{n+1}(X))$. 
    Let $S = X \setminus \choice_{n+1}(X)$ and $S' = X' \setminus \choice_{n+1}(X')$. As $S \subseteq X'$, we only need to show that no elements of $S$ are elements of $\choice_{n+1}(X')$. As $\choice_{n+1}$ is substitutable, $\choice_{n+1}(X') \cap X = \choice_{n+1}(X)$. So, any element of $X$ that is in $\choice_{n+1}(X') $ must be in $\choice_{n+1}(X)$, and therefore cannot be in $S = X \setminus \choice_{n+1}(X)$. Therefore, $S \subseteq S'$.

    Now, we can apply the induction hypothesis, so we assume $\choice^n$ is substitutable to get 

    \begin{align*}
        \choice^n(S') \cap S & \subseteq \choice^n(S) && \text{induction hypothesis} \\
        \choice^n(S') \cap \left(X \setminus \choice_{n+1}(X)\right) & \subseteq \choice^n(S) &&  \text{substitute out $S$}\\      \left(\choice^n(S') \cap X\right) \setminus \left(\choice^n(S') \cap  \choice_{n+1}(X)\right) & \subseteq \choice^n(S) && \text{intersection commutes}\\  
        \choice_{n+1}(X) \cup \left[\left(\choice^n(S') \cap X\right) \setminus \left(\choice^n(S') \cap  \choice_{n+1}(X)\right)\right] & \subseteq \choice_{n+1}(X) \cup \choice^n(S) && \text{add $\choice_{n+1}(X)$} \\ \intertext{Note that $\choice_{n+1}(X) \supseteq\left(\choice^n(S') \cap  \choice_{n+1}(X)\right)$, so the set difference can be simplified out:}
        \choice_{n+1}(X) \cup \left(\choice^n(S') \cap X\right) & \subseteq \choice_{n+1}(X) \cup \choice^n(S) && \text{equivalent} \\ \intertext{Note that for any sets $A,B,C$, $(A \cup B) \cap C = (A \cap C)\cup (B \cap C) \subseteq A \cup (B \cap C)$, so: 
        }
        \left(\choice_{n+1}(X) \cup \choice^n(S') \right) \cap X & \subseteq \choice_{n+1}(X) \cup \choice^n(S) && \text{} \\
        \left(\choice_{n+1}(X) \cup \choice^n(X' \setminus \choice_{n+1}(X')) \right) \cap X & \subseteq \choice_{n+1}(X) \cup \choice^n(X \setminus \choice_{n+1}(X)) &&  \text{substitute out $S$, $S'$}  \\
        \choice^{n+1}(X') \cap X & \subseteq \choice^{n+1}(X) &&  \text{definition of $\choice^{n+1}$}    
    \end{align*}

    Which is the definition of substitutability, completing the proof.
\end{proof}

\begin{restatable}[Additive Variability Bound]{theorem}{addvar}
\label{thm:additive-variability}
If $\choice$ is a sequential composition of $q$-acceptant, $1$-unstable functions with variabilities $m_1, \ldots, m_n$, then $\choice$ has variability at most $\sum_{i=1}^n m_i$.
\end{restatable}

Note that this is only an upper bound; we could trivially construct cases where this is not tight. For example, compose a $\choice$ of two $q$-representative choice functions with the same underlying ordering. On the other hand, it is also easy to construct instances in which this upper bound is tight. Consider the composition of $m$ queues which each rank a student by a different subject area exam. Then the addition of a new student to the applicant pool could displace any of $m$ students, depending on the subject area.

\textit{Intuition}:
The intuition of this proof lies in the fact that ``appending'' a 1-unstable function $\choice_2$ can ``contribute'' $m_2$ variability, and that perturbing the inputs to $\choice$ can only perturb the inputs to $\choice_1$ by the same amount because $\choice_2$ is $1$-unstable.
\begin{proof}
Proof by induction. 

In the base case, for $n = 1$, $\choice(X) = \choice_1(X)$, so $m = m_1$.

For the inductive step, show that $\choice^{n+1}$ has variability $m' \leq m + m_{n+1}$ when comprised of $\choice_{n+1}, \choice_n, ..., \choice_1$, assuming that $\choice^n$ comprised of $\choice_n,...,\choice_1$ has variability $m$ and $\choice_{n+1}$ has variability $m_{n+1}$. 

First, note that $\choice^{n+1}(X) = \choice_{n+1}(X) \cup \choice^n(X \setminus \choice_{n+1}(X))$.
Recall the definition of variability:
\[
m \coloneqq \max_{X\subset \calX} \left|\bigcup\limits_{x'\in\calX} \choice(X) \setminus \choice(X \cup \{x'\})\right|
\]

We are interested in $\bigcup\limits_{x'\in\calX} \choice^{n+1}(X) \setminus \choice^{n+1}(X \cup \{x'\})$ for every $X$. 
For ease of notation, let $X' = X \cup \{x'\}$, $S = X \setminus \choice_{n+1}(X)$, and $S' = X' \setminus \choice_{n+1}(X')$.

\begin{align*}
    \bigcup\limits_{x'\in\calX} \choice^{n+1}(X) \setminus \choice^{n+1}(X') 
    = & \bigcup\limits_{x'\in\calX} \left[\choice_{n+1}(X) \cup \choice^n(S)\right] \setminus \left[\choice_{n+1}(X') \cup \choice^n(S')\right] & \text{notation substitution}\\
    = & \bigcup\limits_{x'\in\calX} \left[\choice_{n+1}(X)  \setminus \choice_{n+1}(X')  \setminus \choice^n(S')\right] \cup \left[\choice^n(S) \setminus \choice_{n+1}(X')  \setminus \choice^n(S')\right]
    & \text{set diff distributes} \\
    = & \bigcup\limits_{x'\in\calX}  \left[\choice_{n+1}(X)  \setminus \choice_{n+1}(X')  \setminus \choice^n(S')\right] \cup \bigcup\limits_{x'\in\calX} \left[\choice^n(S) \setminus \choice_{n+1}(X')  \setminus \choice^n(S')\right] \\
    \subseteq & \bigcup\limits_{x'\in\calX}  \left[\choice_{n+1}(X)  \setminus \choice_{n+1}(X') \right] \cup \bigcup\limits_{x'\in\calX} \choice^n(S)\setminus \choice^n(S') \\        
\end{align*}

So we have 
\begin{equation*}
    \left|\bigcup\limits_{x'\in\calX} \choice^{n+1}(X) \setminus \choice^{n+1}(X \cup \{x'\}) \right| \leq \left|\bigcup\limits_{x'\in\calX}  \choice_{n+1}(X) \setminus \choice_{n+1}(X') \right| + \left| \bigcup\limits_{x'\in\calX} \choice^n(S)\setminus \choice^n(S')\right|
\end{equation*}

See that the left term is bounded precisely by the variability of $\choice_{n+1}$. 
Next, we need to show that the right term is bounded by $m$.
To apply the variability bound to the right term, we need to show that $S' = S \cup \{x''\}$ for some $x''$ for all $S, S'$, which holds if $|S' \setminus S| \leq 1$ and $S \subseteq S'$, or equivalently $S \setminus S' = \emptyset$, for all $x'$.
We will examine $S \setminus S'$ first.
Consider $S \setminus S' = \left[X \setminus \choice_{n+1}(X)\right] \setminus  \left[X' \setminus \choice_{n+1}(X')\right] $. 
\begin{align*}
    S \setminus S' & = \left[X \setminus \choice_{n+1}(X)\right] \setminus  \left[X' \setminus \choice_{n+1}(X')\right]\\
    & = X \setminus \left[X' \setminus \choice_{n+1}(X')\right] \setminus   \choice_{n+1}(X) \\
    & = \left[(X \setminus X') \cup (X \cap \choice_{n+1}(X')) \right] \setminus  \choice_{n+1}(X) & \text{Set difference of difference property}\\
    & = [X \cap \choice_{n+1}(X')] \setminus  \choice_{n+1}(X) & \text{$X \subseteq X'$ by definition}\\        
\end{align*}
Note that this set is always empty if and only if $\choice_{n+1}$ is substitutable, by \Cref{lem:substitutability-term}, so $S \subseteq S'$ when $\choice_{n+1}$ is substitutable.

Next, consider $S' \setminus S = \left[X' \setminus \choice_{n+1}(X')\right] \setminus  \left[X \setminus \choice_{n+1}(X)\right] $. 
\begin{align*}
    S' \setminus S & = \left[X' \setminus \choice_{n+1}(X')\right] \setminus  \left[X \setminus \choice_{n+1}(X)\right] \\
    & = X' \setminus \left[X \setminus \choice_{n+1}(X)\right] \setminus   \choice_{n+1}(X') \\
    & = \left[(X' \setminus X) \cup (X' \cap \choice_{n+1}(X)) \right] \setminus  \choice_{n+1}(X') & \text{Set difference of difference property}\\
    & = \left[\{x'\} \cup (X' \cap \choice_{n+1}(X)) \right] \setminus  \choice_{n+1}(X') & \text{definition of $X'$}\\    
    & = \left[\{x'\} \cup \choice_{n+1}(X)) \right] \setminus  \choice_{n+1}(X') & \text{$\choice_{n+1}(X)) \subseteq X'$}\\           
\end{align*}

So $|S' \setminus S| = \mathbf{1}_{x'\not\in\choice_{n+1}(X')} + |\choice_{n+1}(X)) \setminus  \choice_{n+1}(X')|$. Then there are two cases to consider:
\begin{enumerate}
    \item $x' \not\in \choice_{n+1}(X')$
    \item $x' \in \choice_{n+1}(X')$
\end{enumerate}

Consistency guarantees that $|S' \setminus S| = 1$ in case (1). When $x' \not\in \choice_{n+1}(X')$, $\choice_{n+1}(X)) = \choice_{n+1}(X')$, so $|\choice_{n+1}(X)) \setminus  \choice_{n+1}(X')| = 0$. 

For case (2): $1$-instability means $|X \cap \choice_{n+1}(X') \setminus \choice_{n+1}(X)| + |\choice_{n+1}(X) \setminus \choice_{n+1}(X')| \leq 1$, so it must be that $|\choice_{n+1}(X)) \setminus  \choice_{n+1}(X')| \leq 1$.

So $\left| \bigcup\limits_{x'\in\calX} \choice^n(S)\setminus \choice^n(S')\right| \leq m$ and so $\left|\bigcup\limits_{x'\in\calX} \choice^{n+1}(X) \setminus \choice^{n+1}(X') \right| \leq m_{n+1} + m$
for all $X$, meaning that $m' \leq m_{n+1} + m$, completing the proof.
\end{proof}

\subsubsection{Linear Assignment Problems} \label{sec:lap-proofs}

We highlight \textit{linear assignment problems} as a generalization of sequential composition. 
The linear sum assignment problem can be represented as a bipartite graph matching. 
One one side of the graph are applicants, and on the other are available admissions slots.
Edges connect these nodes according to how well suited each applicant is for each slot.
The admissions decision is made by selecting edges which have maximal sum, subject to the constraint that each applicant and each slot is selected only once.

Consider all edges connecting to slot $s$; then for each element  $x_j \in X$ there is a real number weight $w(x_j, s)$. Then sorting $x \in \calX$ by $w(x,s)$ creates a total ordering, which we call $\succ_s$. Let slot $s'$ be \textit{equivalent} to 
to slot $s$ if and only if $\succ_{s'} = \succ_{s}$. 
We will assume for simplicity that weights are defined such that ties do not occur.
Let assignment $A$ be the set of pairs $\{(x_{s_1}, s_1), ..., (x_{s_q}, s_q)\}$.

\begin{restatable}{lemma}{LAP Ordering} 
\label{lem:lap-ordering}
Let $x_{s}$ denote the element of $\choice(X)$ assigned to slot $s$. Then for all $x \in X$, $x \in \choice(X)$ if $x \succ_s x_s$, and $x_s \succ_s x$ if $x \not\in \choice(X)$.
\end{restatable}

\textit{Intuition}: Linear assignment problems have a structure that resembles a combination of rankings, even though there is not a fixed sequential order.

\begin{proof}
This comes by the optimality of the assignment solution. For contradiction, say there is $x^* \succ_s x_s$ where $x^* \not\in \choice(X)$. Then $w(x^*, s) > w(x_s, s)$, and therefore $\{x^*, s\} \cup A \setminus \{x_s, s\}$ leads to a higher sum than $\choice(X)$, contradicting the fact that $A$ is be optimal.
\end{proof}

\begin{restatable}[Linear Assignment Instability]{theorem}{LAPstable} 
\label{thm:lap-unstable}
Linear assignment problems are 1-unstable.
\end{restatable}

\textit{Intuition}: The ranking-like structure of linear assignment slots maintains substitutability, because each slot maintains a  preference ordering over elements that is independent of the cohort.

\begin{proof}
Recall that for $q$-acceptant functions, substitutability and 1-instability are equivalent by \Cref{thm:sub-unstable-equiv}.

The proof is by contradiction; assume $\choice$ is not substitutable. As per \Cref{lem:unsubstitutable}, there exists $\choice(X), \choice(X')$ where some $x^* \in X$ is not in $\choice(X)$ but $x^* \in \choice(X')$.
By \Cref{lem:lap-ordering}, $\choice(X) = \{x_{s_1}, ..., x_{s_q}\}$ tells us we have $x_{s_i} \succ_{s_i} x^*$ for all $i \in \{1, ..., q\}$. For $x^*$ to be in $\choice(X')$, $x_{s_i}$ must also be in $\choice(X')$, meaning that $\choice(X) \subseteq \choice(X')$. However, as $x^* \in \choice(X')$, $|\choice(X')| > |\choice(X')|$, which violates $q$-acceptance, leading to contradiction.
\end{proof}

\begin{restatable}[Linear Assignment Variability]{theorem}{LAPvariable} 
\label{thm:lap-variability}
The variability of a linear assignment problem is upper bounded by the number of distinct preference orderings induced by its slots. 
\end{restatable}

\textit{Intuition}: 
The proof follows from the observation that, if slot $s$ ranks some other element $x$ in the chosen set higher than the element $x_s$ currently assigned to $s$, $x$ cannot be in the variable set. Then, if $k$ slots (say, slots $s_1, ..., s_k$, with assigned elements $x_{s_1}, ..., x_{s_k}$) induce the same preference ordering over elements, then only the minimal element in $\{x_{s_1}, ..., x_{s_k}\}$ according to that preference ordering can be in $V_\choice(X)$.

\begin{proof}
We want to show that variability is upper bounded by the number of unique orderings. Recall that the variable set $V_\choice(X) \subseteq \choice(X)$ is at most cardinality $q$.

I claim that for any $X$ and $\choice(X)$, $x_{s_i}$ is not in $V_\choice(X)$ if $x_{s_i} \succ_{s_j} x_{s_j}$ for any $j \neq i$. 
To be precise, if $x_{s_i} \in \choice(X)$, and $x_{s_i} \succ_{s_j} x_{s_j}$ for some $j \neq i$, then $x_{s_i} \in \choice(X')$ for any $X' = X \cup \{x'\}$. 

The proof of this claim is by contradiction: assume there exists $\choice(X), \choice(X')$ where some $x_{s_i} \succ_{s_j} x_{s_j}$ but $x_{s_i} \not\in \choice(X')$. Firstly, because of the 1-instability proven in \Cref{thm:lap-unstable}, $|\choice(X) \setminus \choice(X')| = 1$, which means by assumption, $\choice(X) \setminus \choice(X') = \{x_{s_i}\}$, which further implies $x_{s_j} \in \choice(X')$. However,  \Cref{lem:lap-ordering} then applies to $x_{s_j}$, where $x_{s_i} \succ_{s_j} x_{s_j}$ implies $x_{s_i} \in \choice(X')$, leading to contradiction.

Secondly, note that if $s_i, s_j$ induce the same total ordering, then either $x_{s_i} \succ_{s_j} x_{s_j}$ or $x_{s_j} \succ_{s_i} x_{s_i}$ (and the same for $\succ_{s_i}$) for any $\choice(X)$, as a property of any total ordering. Then one of either $x_{s_j}$ or $x_{s_j}$ are not in $V_\choice(X)$. It is intuitive to see that this generalizes: if $s_1, ... s_j$ all induce the same ordering $\succ_{s_j}$, there is exactly one minimal element in $x_{s_1}, ..., x_{s_j}$, because any finite subset of a totally ordered set is well-ordered. 
Say that $x_{s_j}$ is the minimal element, w.l.o.g.; then none of $x_{s_1}, ..., x_{s_{j-1}}$ are in $V_\choice(X)$. More generally, let there be $m$ unique orderings, which we notate $\succ_1,..., \succ_m$. Let $k \in \{1,..., m\}$ and let there be $c_k$ slots that induce $\succ_k$. Then we have 
defined sets of size $c_k-1$ elements that cannot be in $V_\choice(X)$. Then $|V_\choice(X)| \leq m$, and therefore the overall variability is at most $m$. 
\end{proof}

\subsection{Proof of \Cref{thm:variability}} \label{sec:variability-proof}

Recall \Cref{thm:variability}:
\thmvariability*

We see that the second claim comes exactly from \Cref{thm:ranking-variability}, and that combining it with \Cref{thm:additive-variability} directly gives us the first claim.

\subsection{Proof of \Cref{prop:ml-representation}} \label{sec:ml-prop-proof}

Recall \Cref{prop:ml-representation}:
\mlrepresentation*

First consider using an ML model to make decisions; then it can be seen as a choice function with various properties.
Thus we say that an ML model can represent a choice function with various properties if it can have those properties when viewed as a choice function. %

The first statement follows by noting that any rule of the form~\eqref{eq:ml-independent} must either admit or reject an applicant regardless of the applicant pool, so it must be independent and hence $0$-unstable by \Cref{thm:independent-zero-unstable}. %
The second and third statements follow by noting that scores in~\eqref{eq:ml-rank} induce a total order by embedding applicants in $\mathbb R$ and selecting accordingly. Therefore any rule of the form~\eqref{eq:ml-rank} is $q$-representative (by definition) and hence has instability and variability exactly one (by \Cref{thm:ranking-variability}), but no greater.

\subsection{Proof of \Cref{prop:program-variability}} \label{sec:program-prop-proof}

Recall \Cref{prop:program-variability}:
\empiricalmethods*

We construct applicant sets to prove the variability of each empirical method we study, showing that the upper bound from \Cref{thm:variability} based on the number of queues is tight.

Note that both Screened and Open programs are explicitly a ranking, and therefore have variability $1$ by  \Cref{thm:variability}. 
Screened and Open programs with DIA are upper-bounded by variability 2, as they are composed of two queues. It is easy to see that this upper bound is tight. Consider an applicant pool $X$ where the lowest-ranking students in the DIA queue is lower-ranked than the lowest-ranked student in the non-DIA queue. Then $A_\choice(X)$ contains two elements -- adding a new DIA-qualified student might lead to rejecting the lowest-ranked student in the DIA queue, while a non-DIA student would lead to a change in the non-DIA queue. Similarly, in the three-queue Ed. Opt. case, the variability $3$ bound is tight, as there are applicant pools where adding a single high-, middle-, or low-category student would displace a different student in the corresponding queue. Likewise, the six queues in Ed. Opt. with DIA programs lead to six diffrent students who might be displaced depending on what category a new applicant falls into.

\section{Additional Empirical Setting Details}
\subsection{Feature Engineering} \label{sec:feature-engineering-appendix}
In order to isolate the effects of choice function characteristics on learnability, we engineer the input features to the model to avoid spurious correlations. In addition to tailoring each model's inputs to remove irrelevant features (e.g., Ed. Opt. category features for an Open program), we input priority flags directly as features (i.e, a flag for ``priority 1'' instead of a categorical residence borough), and pre-calculate students' Screened grade tiers. Specifically, we feature engineer Screened grade tiers according to the 2025 (current at time of writing) breakdown, which determines cutoffs based on percentile thresholds of the citywide student grade distribution (top 15\%, 30\%, and 50\% of all students).

\subsection{Changes in 2021 and 2022 Admissions Methods}
Note that Screened programs in 2021 did not follow a wholly uniform system; some programs ranked students by raw grades without any grouping into tiers, and other programs determined their own cutoffs for grade tiers. This both motivates our usage of a simulator --- to hold the choice functions at play constant, even though they historically were modified --- as well as explains the differential in simulator accuracy for 2021.

\section{Experiment Details and Validations} \label{sec:empirics-appendix}

\subsection{Synthetic Admissions Method Details} \label{sec:synthetics-appendix}
As mentioned in the main text, we create three synthetic methods for comparisons of instability, two of which are $0$-unstable and one of which is $5$-unstable.

The first, the Uniform Threshold, is a threshold on tiebreaker number, which we normalize to a [0, 1] distribution globally. A student is admitted if their tiebreaker number is less than $0.5$ (i.e, above the 50\% percentile, globally), and rejected otherwise. 

The second, the Conditional Threshold, is largely similar, but with different thresholds depending on a student's Ed. Opt. category. A High-category student is admitted with a tiebreaker number below $0.7$, Middle-category students $0.5$, and Low-category students $0.3$. This serves as a slightly less trivial to learn, but still simple and $0$-unstable rule to learn as a control.

The $5$-unstable function is more complex. The base intuition here is that, drawing on other examples of complement effects in higher-$d$-unstable functions, a hypothetical program wants to pair exactly equal numbers of High- and Low-category students. However, to ensure that different samples with high $r_\choice$ occur, the number of seats reserved for such complements is a periodic function of the number of the High- and Low-category students that apply. This provides meaningful fluctuations, and the change in seats reserved is bounded to be at most two per student. The function is $5$-unstable and not $3$-unstable, because in the worst case, adding a new student \textit{decreases} the number of paired spots, even as the new student displaces a previously accepted student.

\subsection{Covariate Shift in Screened Programs} \label{sec:screened-appendix}
While we account for many other complications to training a machine learning model, not all effects can be fully controlled. The distribution of Screened grade tiers in 2022 was noticeably skewed; an outright majority of students were placed in the highest tier based on the grading system used, and a majority of students had the highest possible grade even prior to tier grouping. This is a severe covariate shift. As a result, model training results for Screened/Screened with DIA programs fail dramatically when evaluated out of distribution in 2022. In the following experimental validations (see \Cref{sec:argmax-appendix} and \Cref{sec:models-appendix}), Screened performance drops when comparing individual admissions methods, and aggregate performance for $1$-unstable and $1$- or $2$-variable functions drops well past what would be expected.

However, this covariate shift further highlights the foundation of our work: cohort effects \textit{matter}. Even the simplest and most ideal class of capacity-constrained choice functions can cause drastic changes in performance when samples are drawn from a different distribution. In the language of distribution shift, our work shows that choice functions induce structured changes in the joint distribution between features and labels; these structured changes are partially \textit{driven by} covariate shift, and conventional models of covariate shift would not explain why grade inflation makes admissions prediction so difficult.

\subsection{Disaggregated Performance By Method} 
\label{sec:per-method-appendix}

We provide a disaggregated plot (\Cref{fig:method}) of model accuracies for each individual admissions method, without grouping by instability or variability. Train and test subset performance is left unplotted for visual clarity.

\begin{figure}[tbh]
\centering
\includegraphics[width=\textwidth]{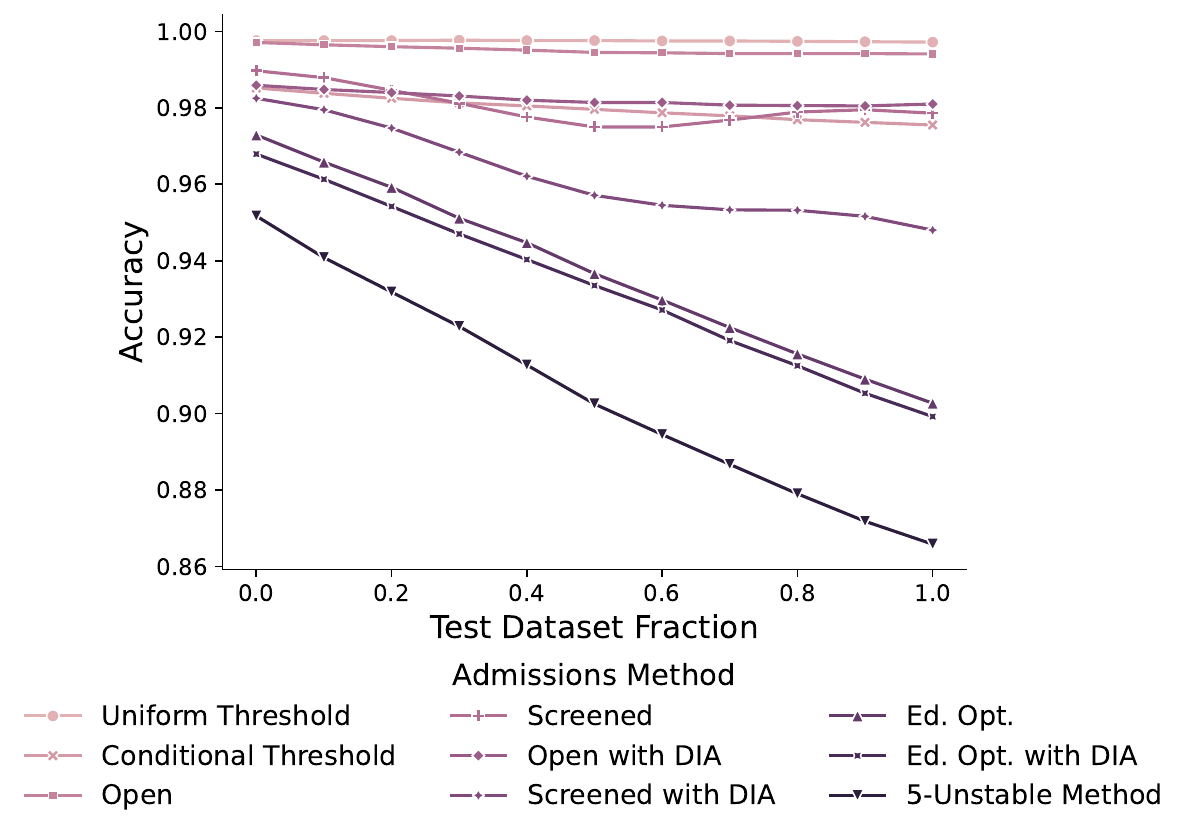}
\caption{Model accuracy for each admissions method, under increasing levels of distribution shift.}
\label{fig:method}
\end{figure}

\subsection{Results with Independent Discretization}
\label{sec:argmax-appendix}

We present our experiment results in \Cref{fig:argmax} using independent discretization; that is, as per \Cref{eq:ml-independent}, with $t = 1$. Note that performance drops across the board, but differentially; as explained in \Cref{sec:screened-appendix}, any independent threshold for Screened programs is strongly affected by grade inflation. However, as seen in \Cref{fig:argmax-method}, all non-Screened (or Screened with DIA) methods have performance roughly sorted by instability and variability, in line with our theory.

\begin{figure}[tbh]
\centering
\begin{subfigure}[b]{0.45\textwidth}
    \centering
\includegraphics[width=\textwidth]{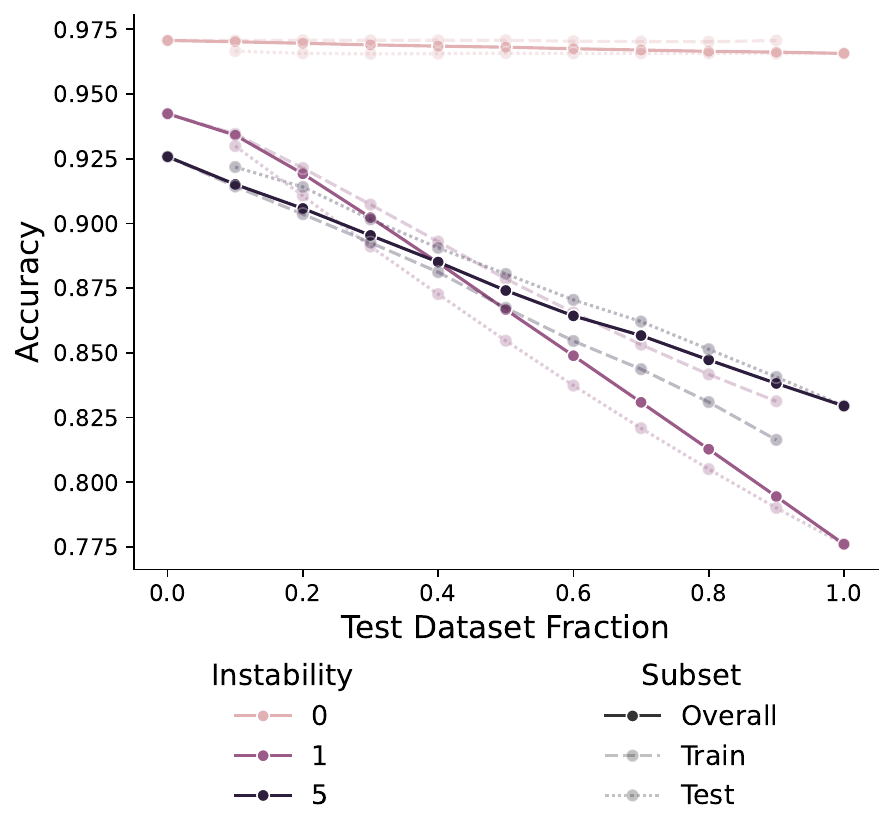}
\caption{Replicating \Cref{fig:instability} with independent discretization.}
\label{fig:argmax-instability}
\end{subfigure}
\hfill
\begin{subfigure}[b]{0.45\textwidth}
    \centering
\includegraphics[width=\textwidth]{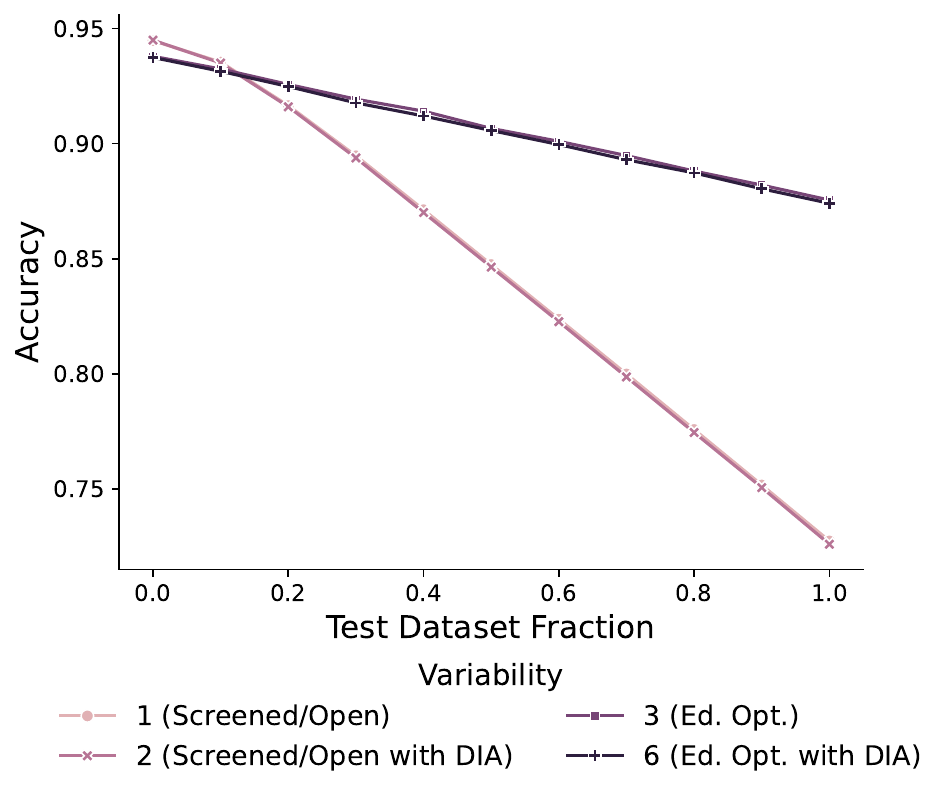}
\caption{Replicating \Cref{fig:variability} with independent discretization.}
\label{fig:argmax-variability}
\end{subfigure}
\hfill
\begin{subfigure}[b]{0.6\textwidth}
    \centering
\includegraphics[width=\textwidth]{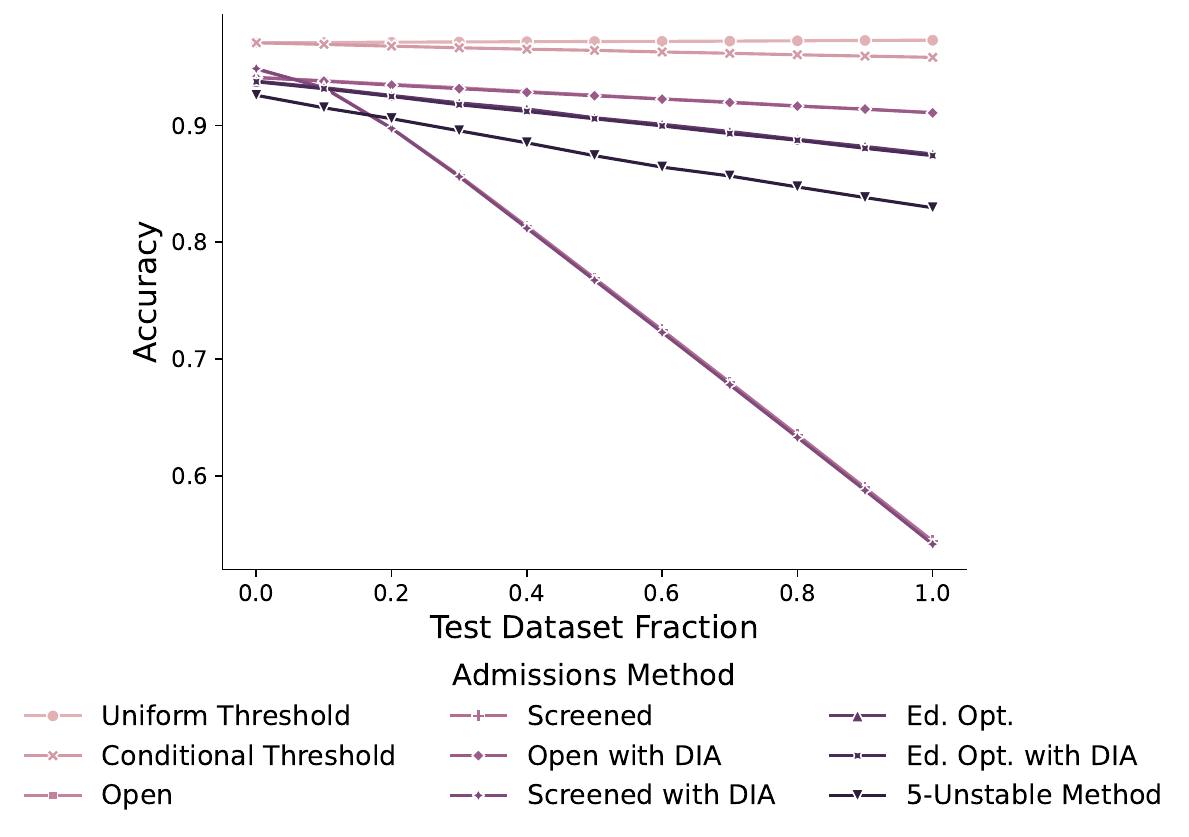}
\caption{Replicating \Cref{fig:method} with independent discretization.}
\label{fig:argmax-method}
\end{subfigure}
\caption{Replication of empirical performance figures using independent discretization.}
\label{fig:argmax}
\end{figure}

\subsection{Results with Different Model Classes}
\label{sec:models-appendix}

In \Cref{fig:xgb}, we present validation results using a different model class. Recall that our main text results use a simple, $L_2$-regularized logistic regression. To demonstrate that the effects of distribution shift we replicate our results using a more expressive model class. \Cref{fig:xgb} shows our replication results using extreme gradient boosting, as implemented in the package ``xgboost'', with engineered monotonic constraints for individual features where applicable (e.g., a higher Screened grade tier is strictly better). Note that these constraints are \textit{weakly} monotonic; predicted scores can either increase \textit{or remain constant} with monotonic features.

Note that, despite its simplicity, logistic regression has a particularly important advantage in our setting: its weights induce a strictly monotonic structure to features. While this is often a limitation in more complex supervised learning settings with higher-dimensional features, it proves most compatible with our probability ranking approach. A student with a tiebreaker number only marginally better than that of another student still ought to be ranked ahead of the other, whereas model classes that use structures such as decision trees or halfspaces often lack such highly fine-grained and strict monotonicity. As a result, logistic regression produces results robust to the distribution shift in Screened grade tiers, as seen in the main text, while other methods become brittle.

\begin{figure*}[tbh]
\centering
\begin{subfigure}[b]{0.45\textwidth}
    \centering
\includegraphics[width=\textwidth]{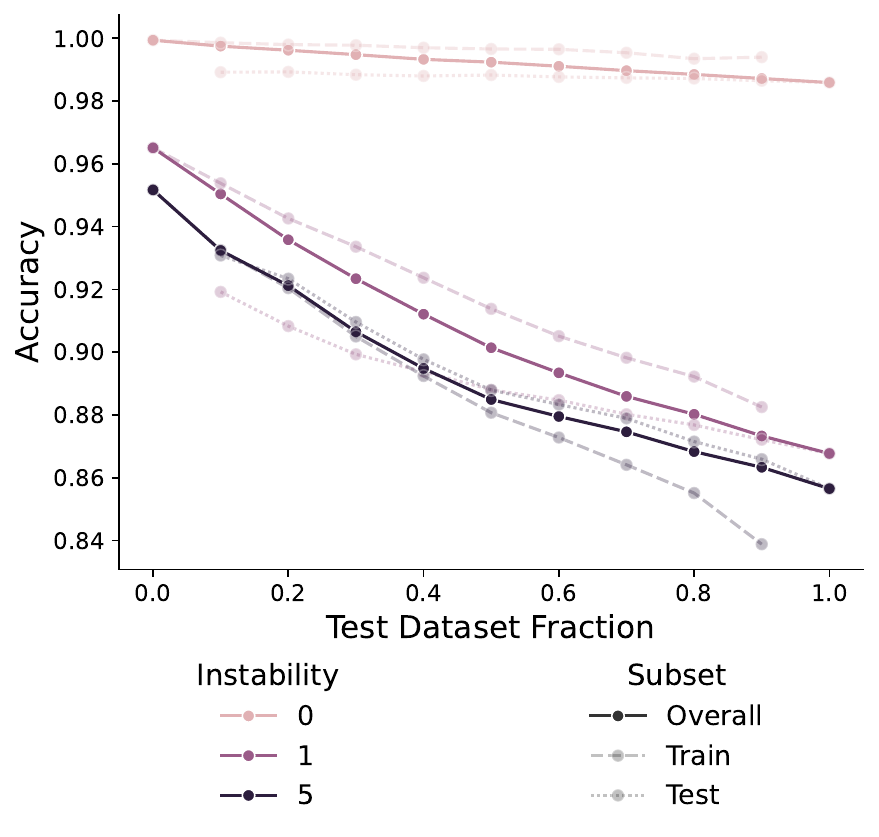}
\caption{Replicating \Cref{fig:instability} with extreme gradient boosting models.}
\label{fig:xgb-instability}
\end{subfigure}
\hfill
\begin{subfigure}[b]{0.45\textwidth}
    \centering
\includegraphics[width=\textwidth]{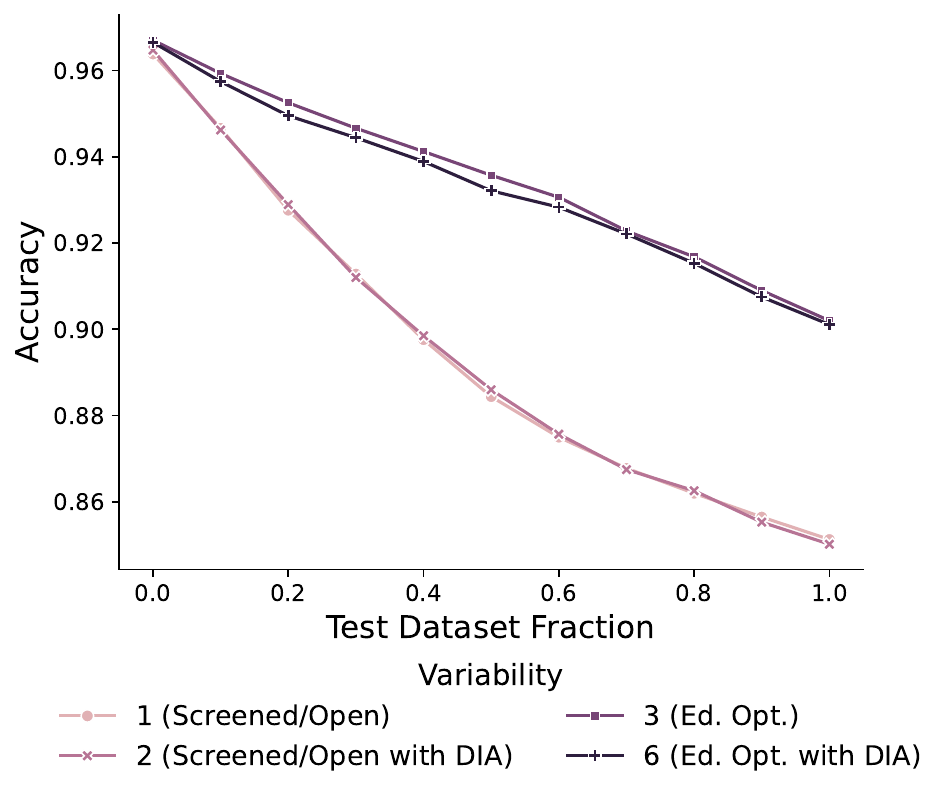}
\caption{Replicating \Cref{fig:variability} with extreme gradient boosting models.}
\label{fig:xgb-variability}
\end{subfigure}
\hfill
\begin{subfigure}[b]{0.6\textwidth}
    \centering
\includegraphics[width=\textwidth]{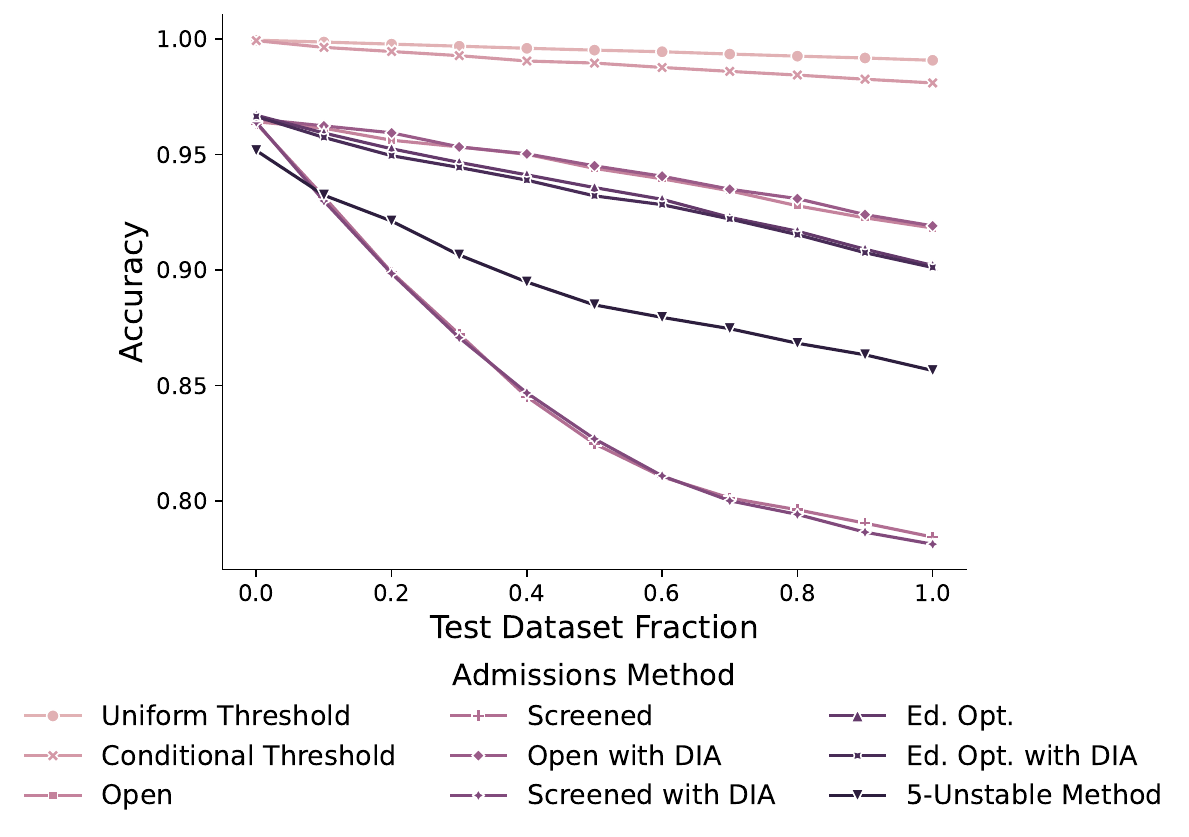}
\caption{Replicating \Cref{fig:method} with extreme gradient boosting models.}
\label{fig:xgb-method}
\end{subfigure}
\caption{Replication of performance figures with extreme gradient boosting models.}
\label{fig:xgb}
\end{figure*}

\subsection{Results with Flipped Training and Testing Years}
\label{sec:rev-appendix}

In this section, we present validation results when using the 2022-2023 admissions year as the training distribution and using the 2021-2022 admissions year as the test distribution. Due to the distribution shift between years, and asymmetries in difficulty learning, it is possible that models trained on one year may generalize to the other year better than in the other direction. We show in \Cref{fig:rev} that all results qualitatively replicate when reversing the training and testing years.

\begin{figure*}[tbh]
\centering
\begin{subfigure}[b]{0.45\textwidth}
    \centering
\includegraphics[width=\textwidth]{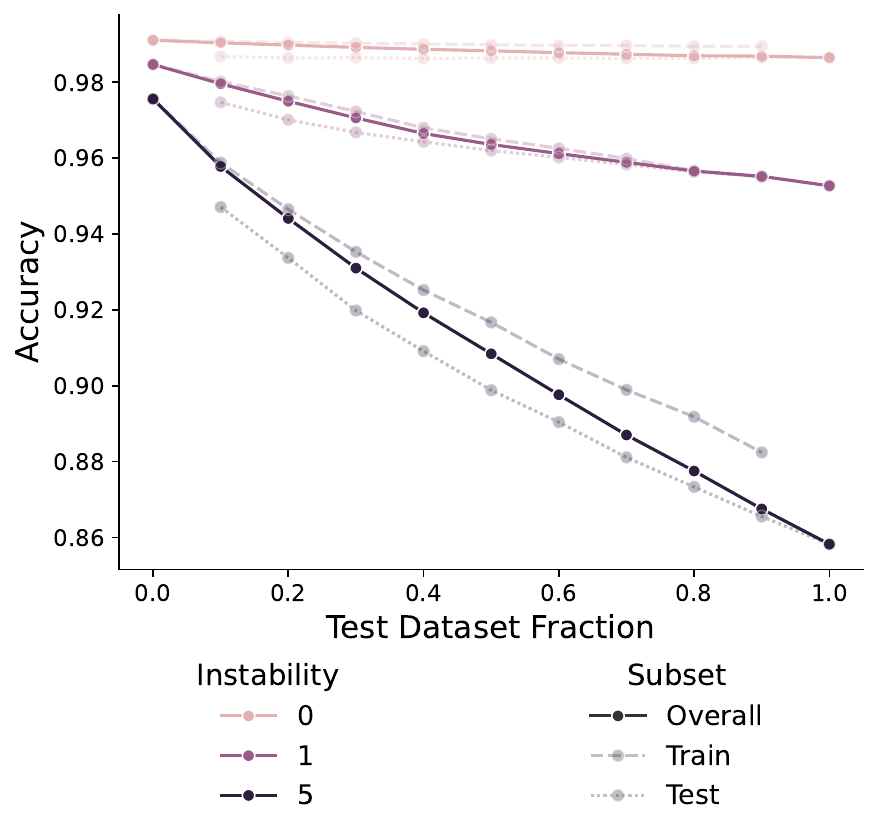}
\caption{Replicating \Cref{fig:instability} with reversed train and test years.}
\label{fig:rev-instability}
\end{subfigure}
\hfill
\begin{subfigure}[b]{0.45\textwidth}
    \centering
\includegraphics[width=\textwidth]{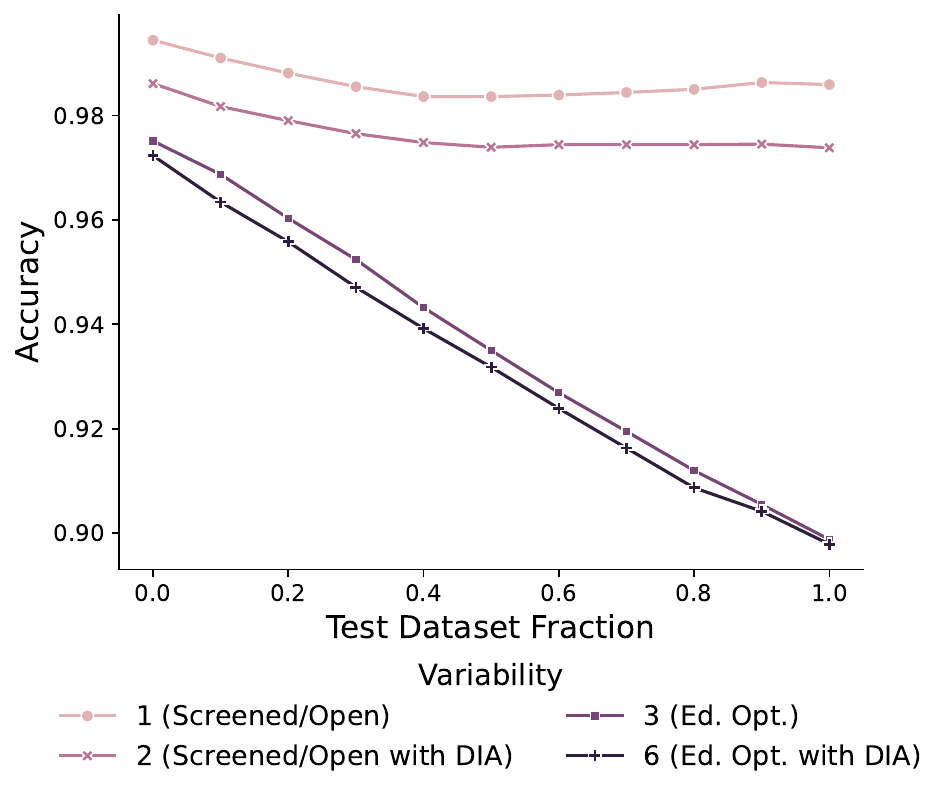}
\caption{Replicating \Cref{fig:variability} with reversed train and test years.}
\label{fig:rev-variability}
\end{subfigure}
\hfill
\begin{subfigure}[b]{0.6\textwidth}
    \centering
\includegraphics[width=\textwidth]{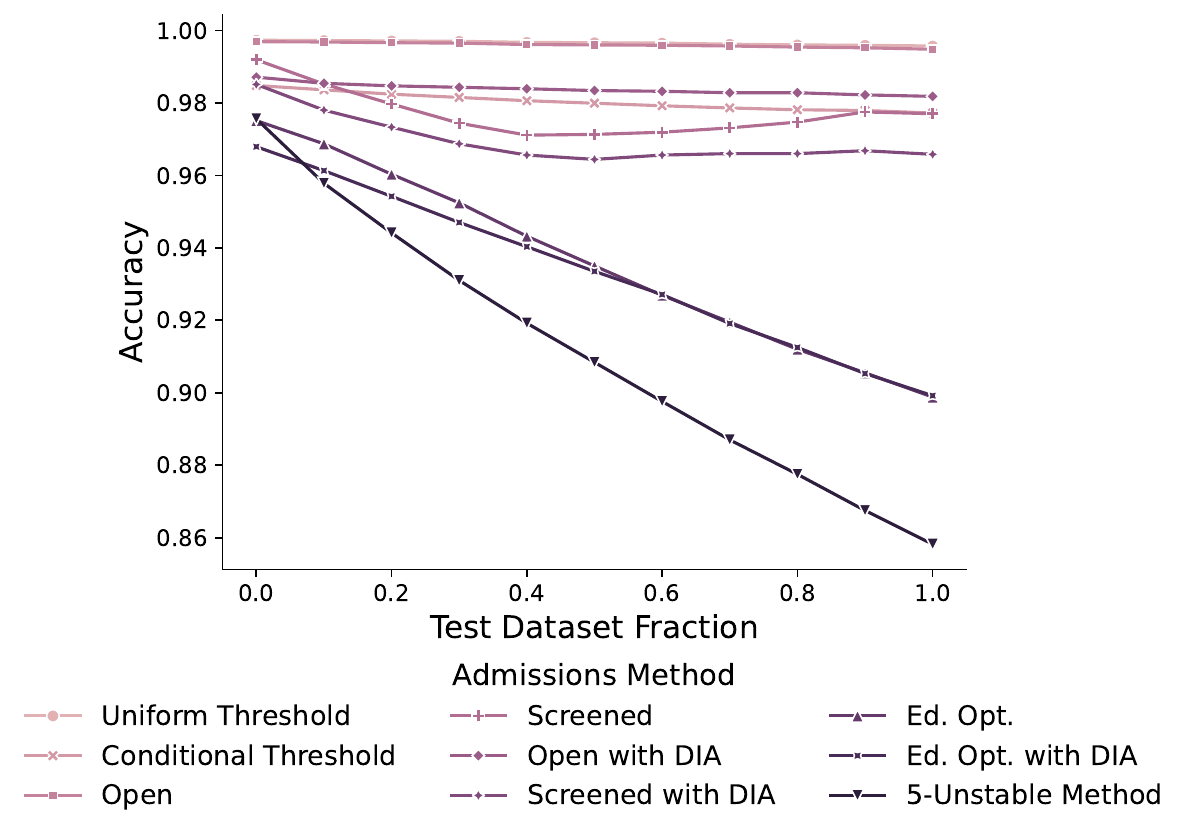}
\caption{Replicating \Cref{fig:method} with reversed train and test years.}
\label{fig:rev-method}
\end{subfigure}
\caption{Replication of performance figures with reversed train and test years.}
\label{fig:rev}
\end{figure*}

\subsection{Results with Fixed Dataset Sizes and Capacity}
\label{sec:fixed-appendix}

In this section, we present validation results when fixing the size and capacity of every interpolated dataset to be exactly the same. We show that all results qualitatively replicate in \Cref{fig:fixed}.

\begin{figure*}[tbh]
\centering
\begin{subfigure}[b]{0.45\textwidth}
    \centering
\includegraphics[width=\textwidth]{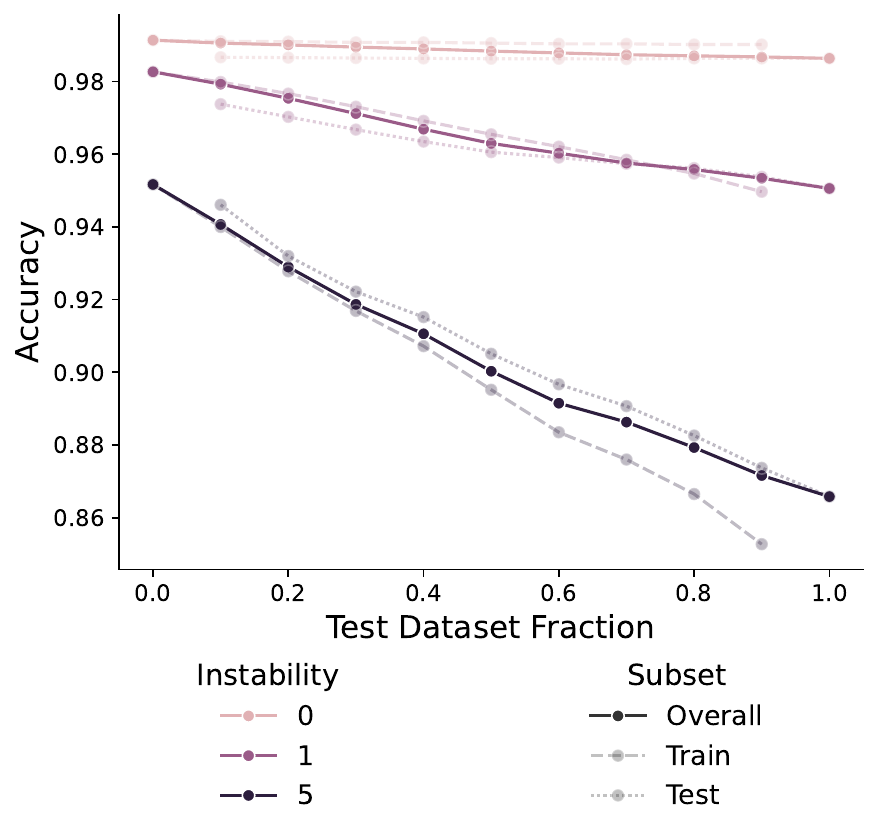}
\caption{Replicating \Cref{fig:instability} with fixed dataset and capacity sizes.}
\label{fig:fixed-instability}
\end{subfigure}
\hfill
\begin{subfigure}[b]{0.45\textwidth}
    \centering
\includegraphics[width=\textwidth]{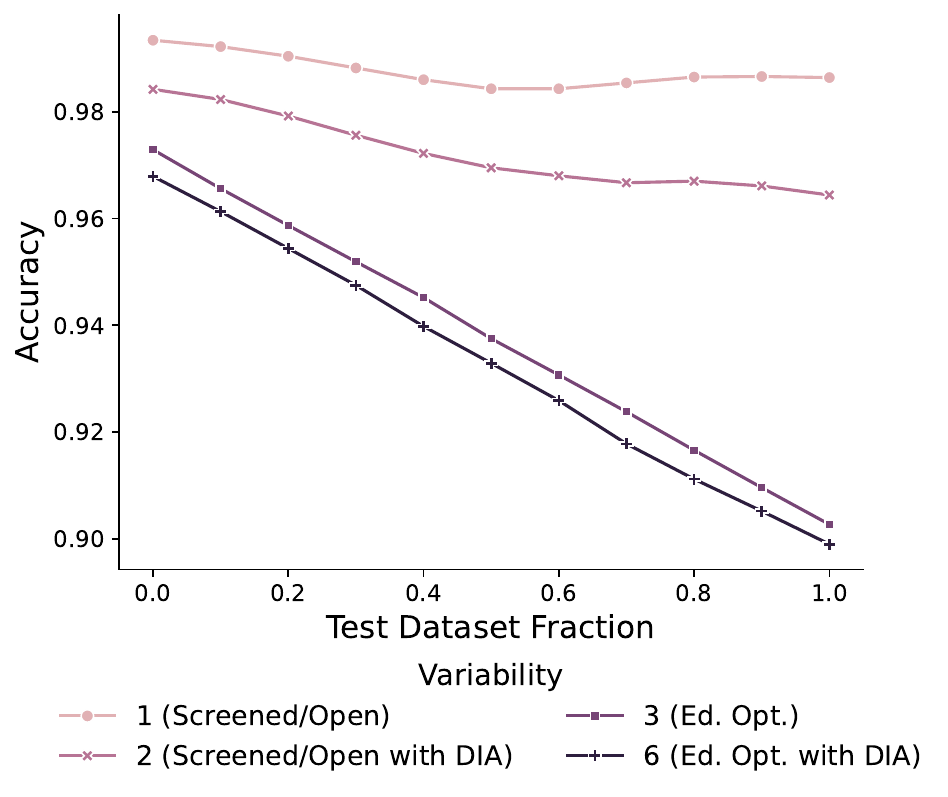}
\caption{Replicating \Cref{fig:variability} with dataset and capacity sizes.}
\label{fig:fixed-variability}
\end{subfigure}
\hfill
\begin{subfigure}[b]{0.6\textwidth}
    \centering
\includegraphics[width=\textwidth]{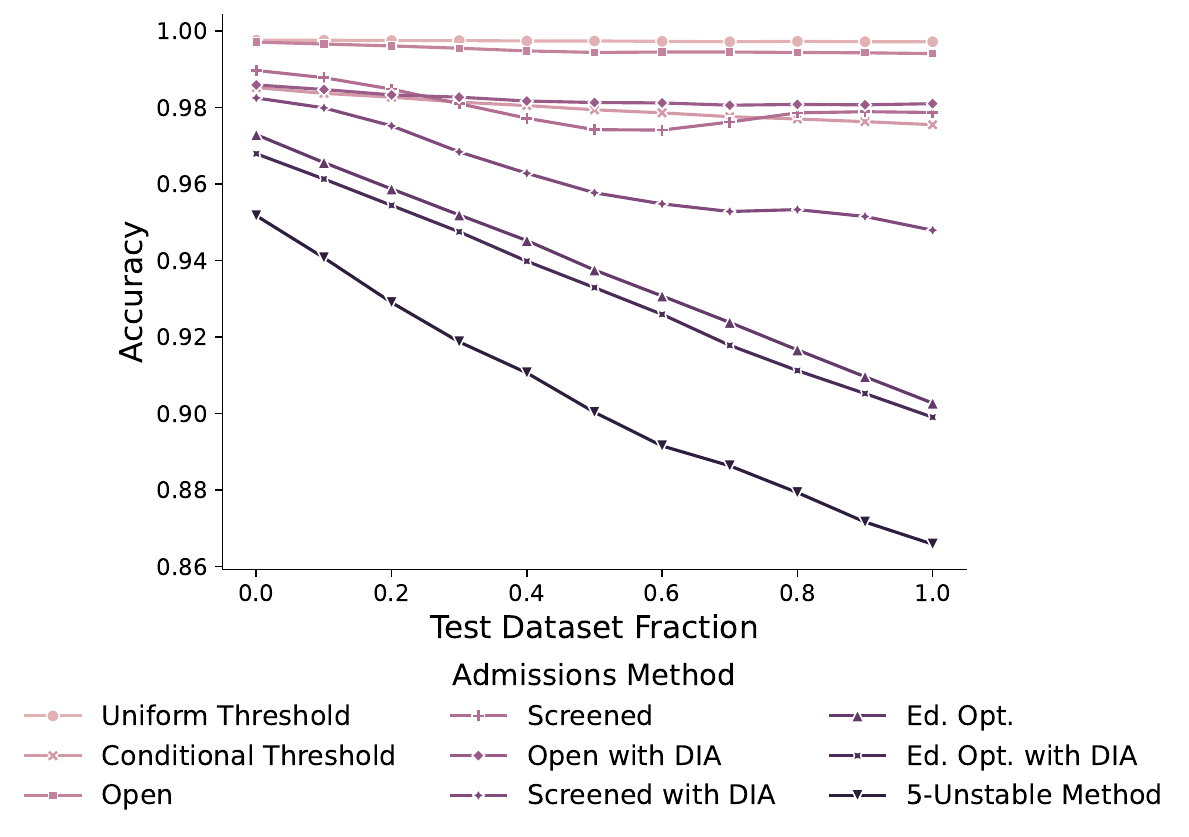}
\caption{Replicating \Cref{fig:method} with dataset and capacity sizes.}
\label{fig:fixed-method}
\end{subfigure}
\caption{Replication of performance figures with dataset and capacity sizes.}
\label{fig:fixed}
\end{figure*}

\subsection{Results without Borough or Continuing Student Priority}
\label{sec:nopriority-appendix}

Our theoretical and empirical analyses focus on the impact of instability and variability. However, the real admissions process is more complex, and our simulator also captures a variety of preferences that do not affect the instability or variability of an admissions choice function --- namely, borough and continuing student priorities. For this reason, we replicate our results using a simplified form of our simulator that ignores these priority preferences. \Cref{fig:nopriority} presents the results of our experiment, showing that more unstable and variable functions become harder to learn, while simpler admissions methods become easier.

\begin{figure*}[tbh]
\centering
\begin{subfigure}[b]{0.45\textwidth}
    \centering
\includegraphics[width=\textwidth]{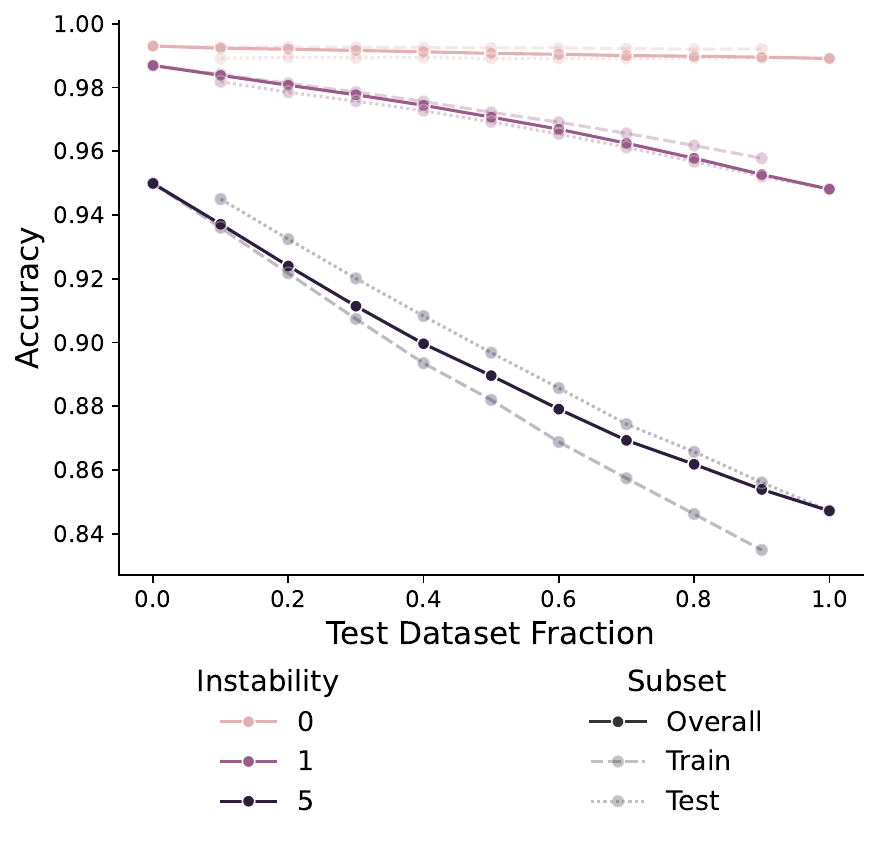}
\caption{Replicating \Cref{fig:instability} with borough and continuing student priority excluded.}
\label{fig:nopriority-instability}
\end{subfigure}
\hfill
\begin{subfigure}[b]{0.45\textwidth}
    \centering
\includegraphics[width=\textwidth]{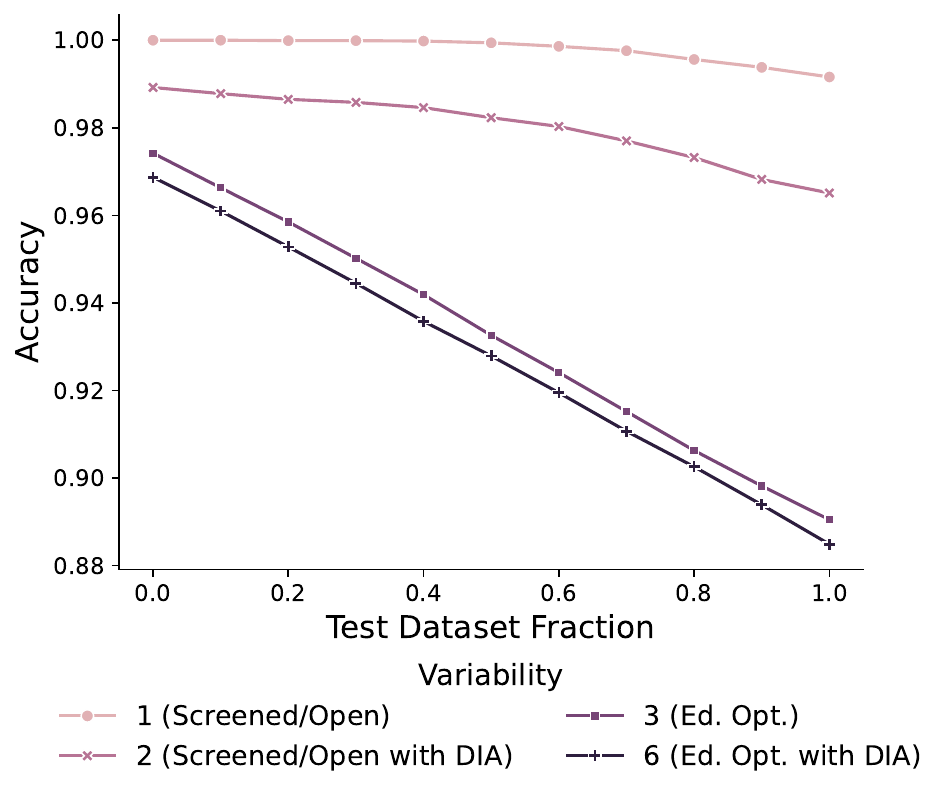}
\caption{Replicating \Cref{fig:variability} with borough and continuing student priority excluded.}
\label{fig:nopriority-variability}
\end{subfigure}
\hfill
\begin{subfigure}[b]{0.6\textwidth}
    \centering
\includegraphics[width=\textwidth]{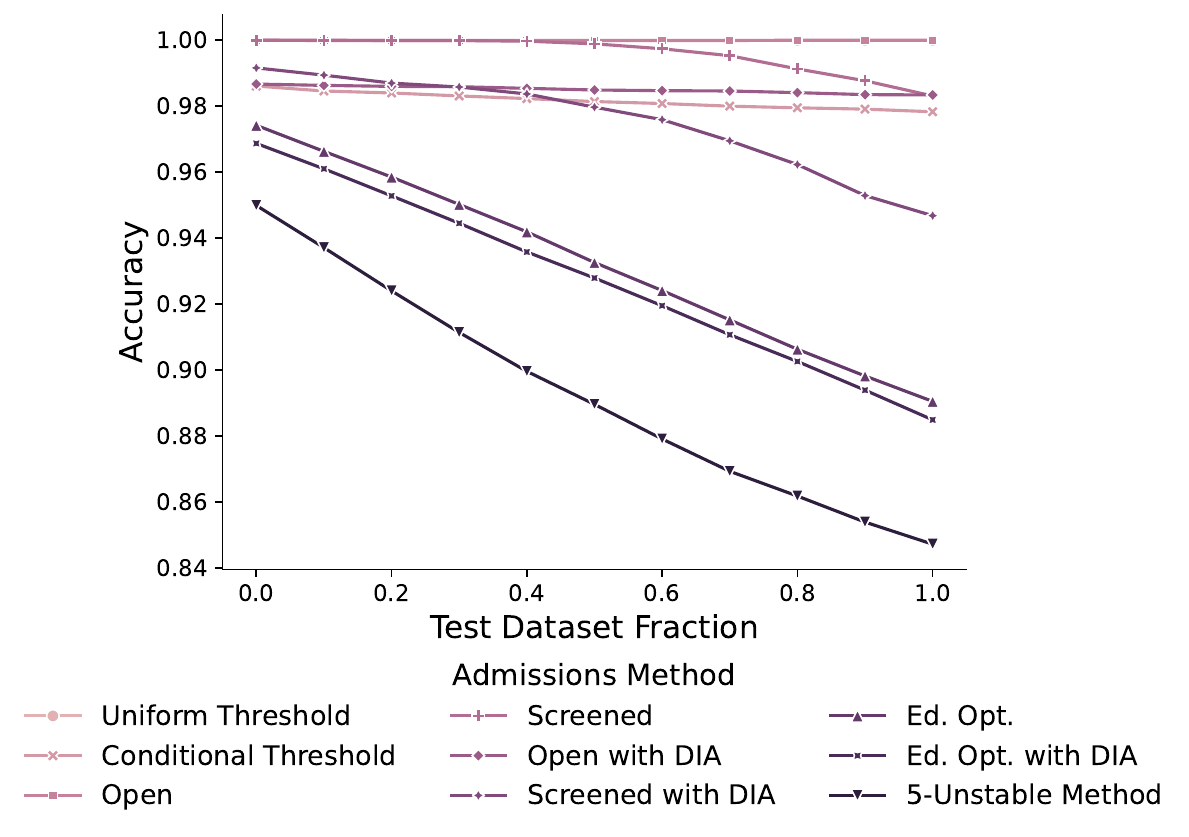}
\caption{Replicating \Cref{fig:method} with borough and continuing student priority excluded.}
\label{fig:nopriority-method}
\end{subfigure}
\caption{Replication of performance figures with borough and continuing student priority excluded.}
\label{fig:nopriority}
\end{figure*}

\subsection{Publicly Available Synthetic Dataset}\label{sec:synth-appendix}

While the original data in our experiments cannot be shared, we are able to provide a synthetic data generator based on aggregated applicant distributions. This generator is included with our replication code at \url{https://github.com/evan-dong/admissions-prediction}.

Our data generator jointly generates three core features most correlated and critical to admission: grade tiers, Ed. Opt. categories, and DIA-qualification. These distributions are conditioned on the combination of admissions method (Open, Screened, and Ed. Opt.) and DIA participation (with or without DIA participation); for every such combination, we aggregate applicant pools across all programs of this type and add noise to create the base distribution. For schools, we draw on the 2021 NYC high school program directory, which is publicly available via the NYC open data portal. \Cref{fig:synth}shows that our results broadly replicate in this synthetic setting. 

\begin{figure*}[tbh]
\centering
\begin{subfigure}[b]{0.45\textwidth}
    \centering
\includegraphics[width=\textwidth]{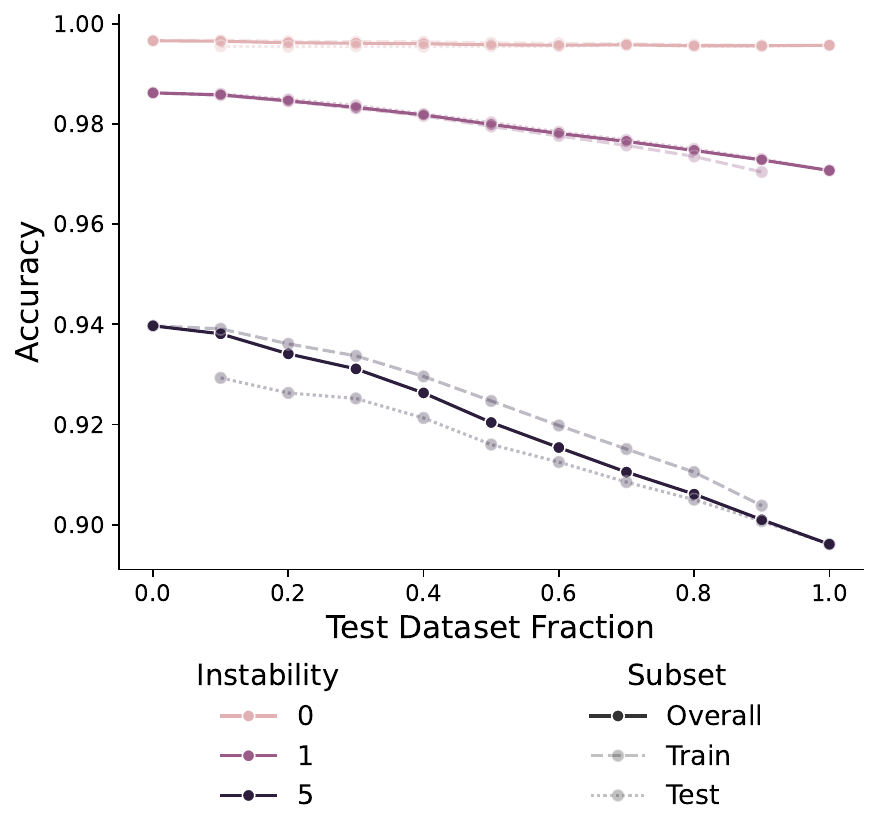}
\caption{Replicating \Cref{fig:instability} with synthetically generated data.}
\label{fig:synth-instability}
\end{subfigure}
\hfill
\begin{subfigure}[b]{0.45\textwidth}
    \centering
\includegraphics[width=\textwidth]{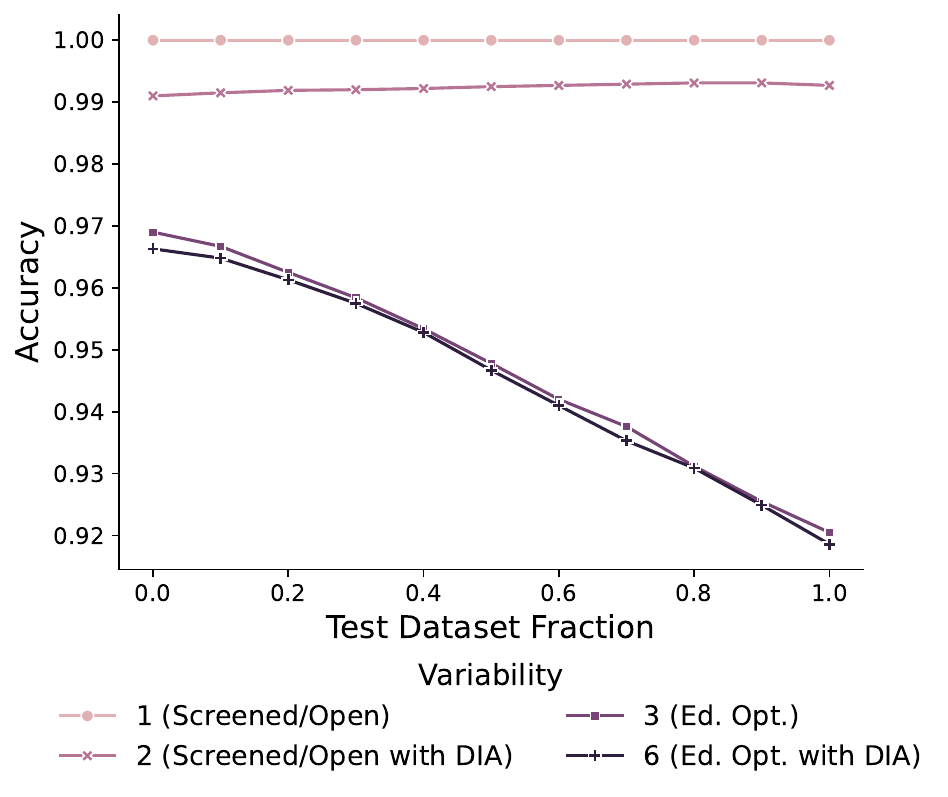}
\caption{Replicating \Cref{fig:variability} with synthetically generated data.}
\label{fig:synth-variability}
\end{subfigure}
\hfill
\begin{subfigure}[b]{0.6\textwidth}
    \centering
\includegraphics[width=\textwidth]{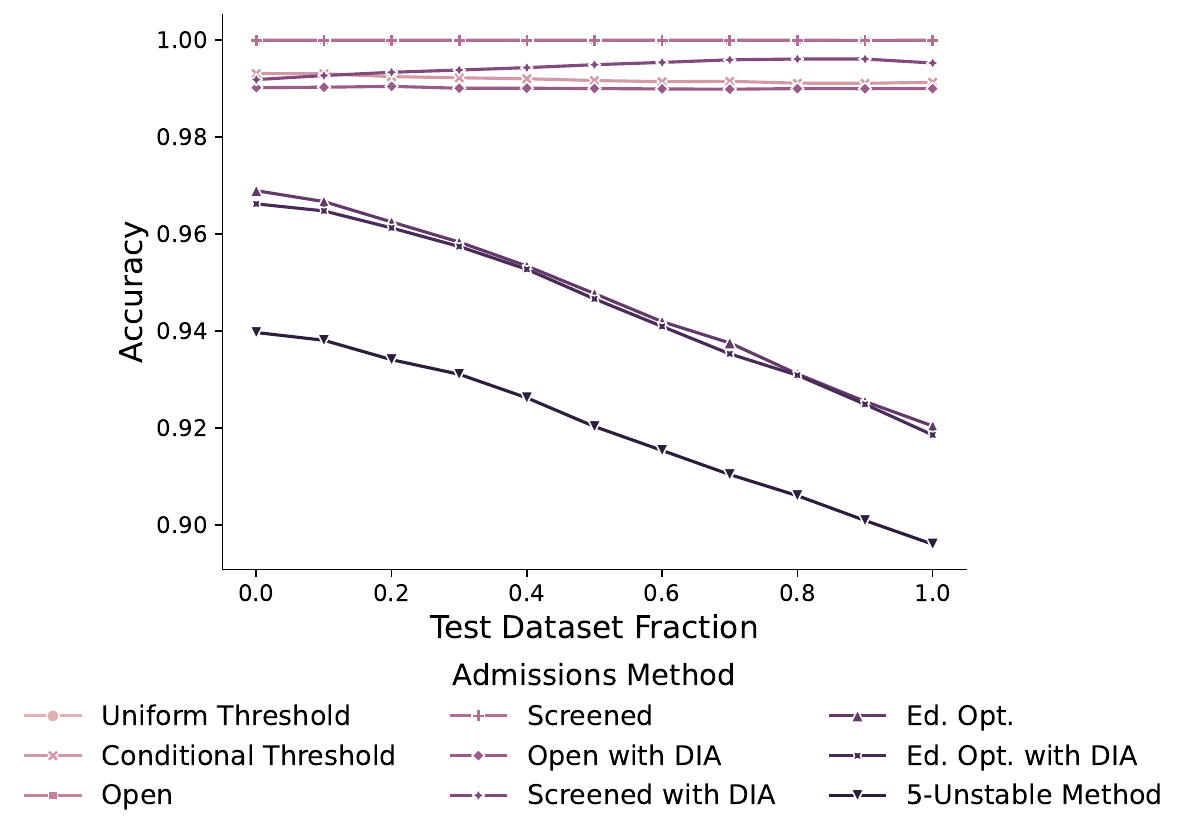}
\caption{Replicating \Cref{fig:method} with synthetically generated data.}
\label{fig:synth-method}
\end{subfigure}
\caption{Replication of performance figures with synthetically generated data.}
\label{fig:synth}
\end{figure*}

\subsection{Results with In-Distribution Data}
\label{sec:id-appendix}

Using the synthetic data generator introduced in \Cref{sec:synth-appendix}, we are able to generate multiple samples of in-distribution data. \Cref{fig:id} presents the results of our experiment when the training and testing data are drawn from the same distribution. Note that higher-variability and higher-instability choice functions still experience substantial performance decay.

\begin{figure*}[tbh]
\centering
\begin{subfigure}[b]{0.45\textwidth}
    \centering
\includegraphics[width=\textwidth]{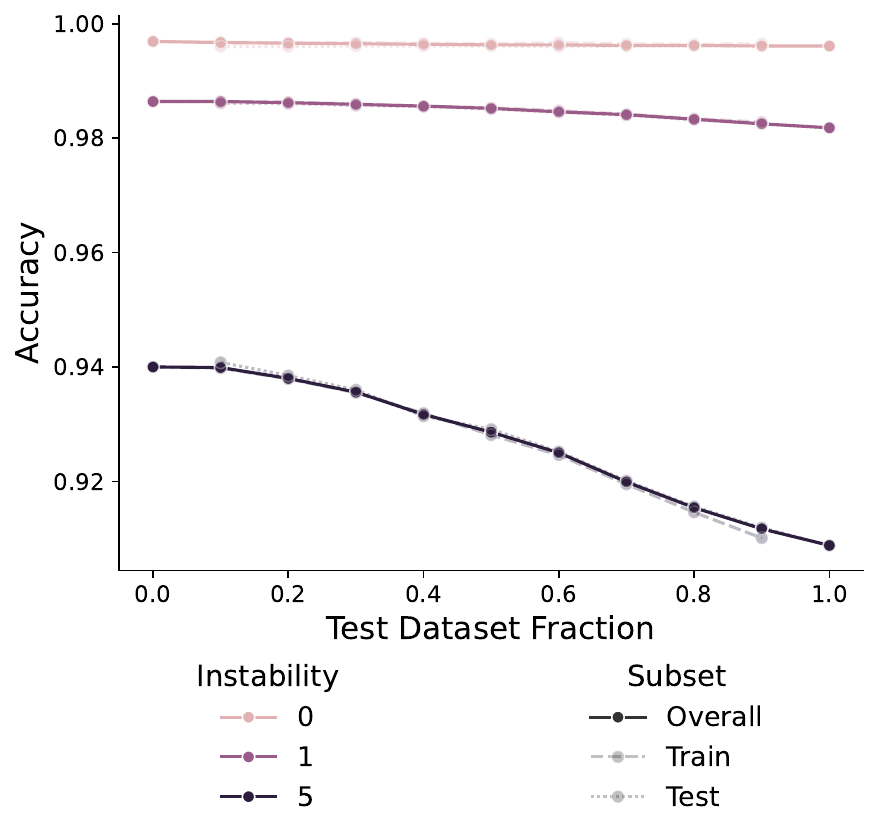}
\caption{Replicating \Cref{fig:instability} tested on in-distribution synthetically generated data.}
\label{fig:id-instability}
\end{subfigure}
\hfill
\begin{subfigure}[b]{0.45\textwidth}
    \centering
\includegraphics[width=\textwidth]{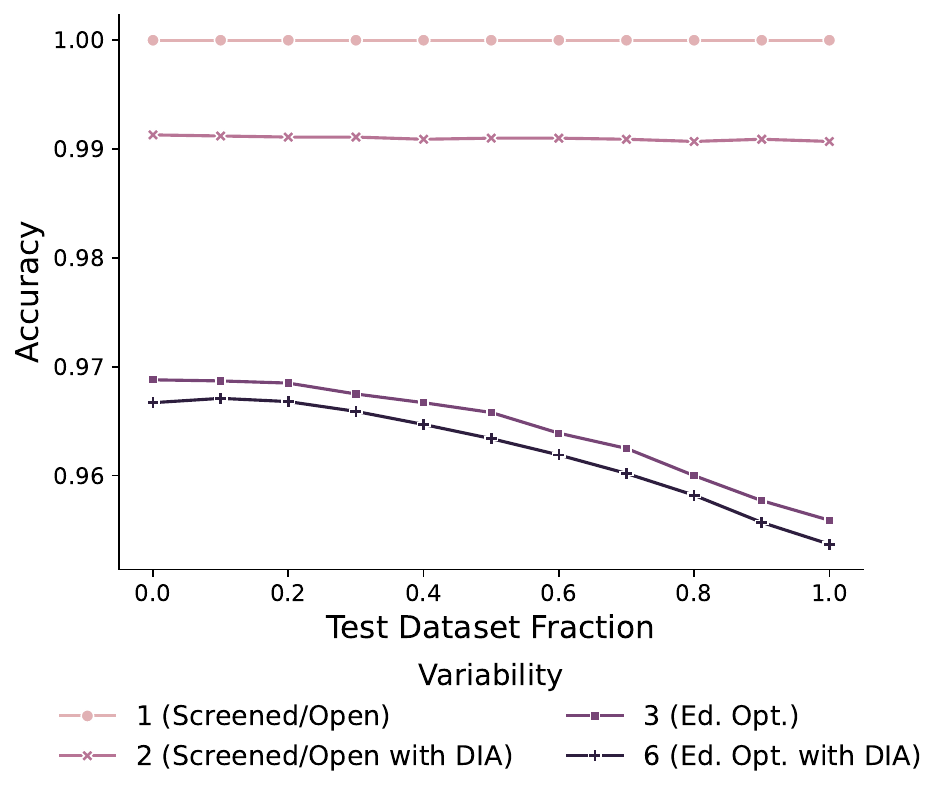}
\caption{Replicating \Cref{fig:variability} tested on in-distribution synthetically generated data.}
\label{fig:id-variability}
\end{subfigure}
\hfill
\begin{subfigure}[b]{0.6\textwidth}
    \centering
\includegraphics[width=\textwidth]{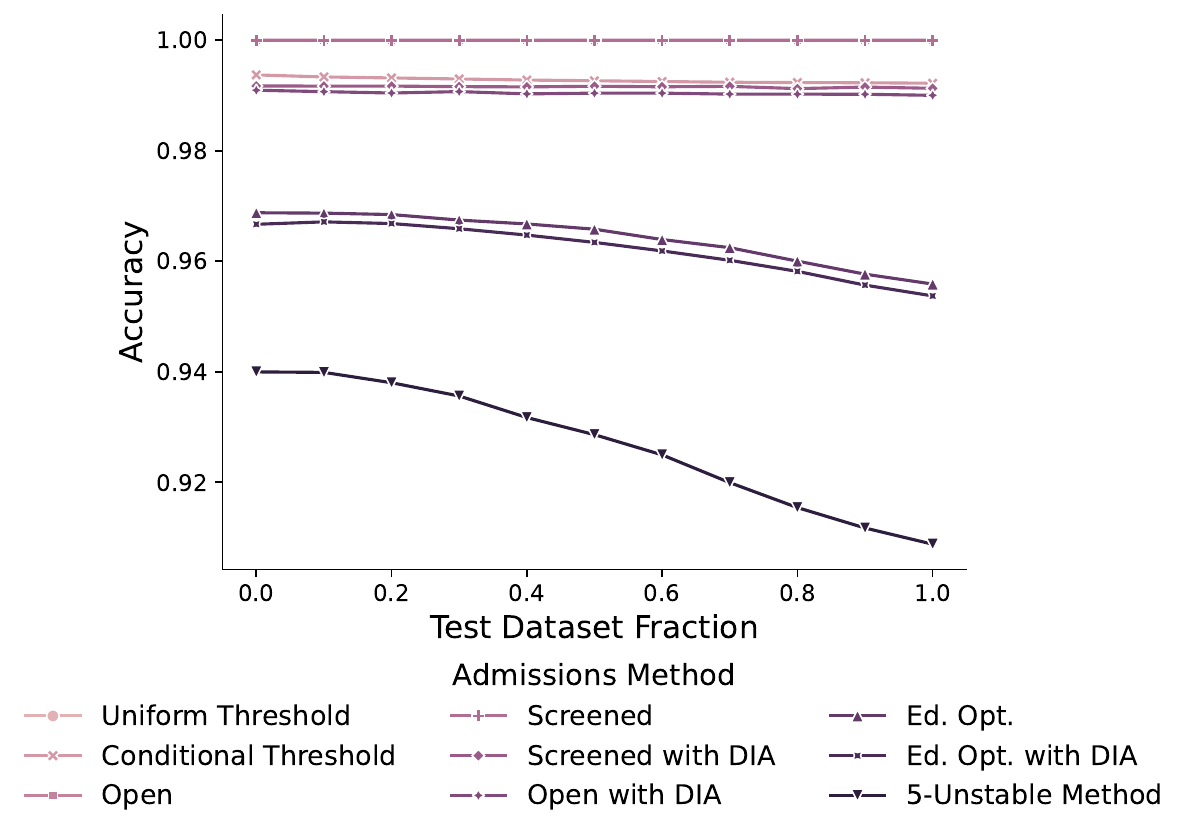}
\caption{Replicating \Cref{fig:method} tested on in-distribution synthetically generated data.}
\label{fig:id-method}
\end{subfigure}
\caption{Replication of performance figures tested on in-distribution synthetically generated data.}
\label{fig:id}
\end{figure*}

%% file: main.bbl
\begin{thebibliography}{47}
\providecommand{\natexlab}[1]{#1}

\bibitem[{Arnosti, Bonet, and Sethuraman(2024)}]{arnosti2024explainable}
Arnosti, N.; Bonet, C.; and Sethuraman, J. 2024.
\newblock Explainable affirmative action.
\newblock In \emph{Proceedings of the 25th ACM Conference on Economics and Computation}, 310--310.

\bibitem[{Arrow(1959)}]{arrow1959rational}
Arrow, K.~J. 1959.
\newblock Rational choice functions and orderings.
\newblock \emph{Economica}, 26(102): 121--127.

\bibitem[{Bommasani et~al.(2022)Bommasani, Creel, Kumar, Jurafsky, and Liang}]{bommasani2022pickingpersondoesalgorithmic}
Bommasani, R.; Creel, K.~A.; Kumar, A.; Jurafsky, D.; and Liang, P.~S. 2022.
\newblock Picking on the same person: Does algorithmic monoculture lead to outcome homogenization?
\newblock \emph{Advances in Neural Information Processing Systems}, 35: 3663--3678.

\bibitem[{Bruggink and Gambhir(1996)}]{bruggink1996statistical}
Bruggink, T.~H.; and Gambhir, V. 1996.
\newblock Statistical models for college admission and enrollment: A case study for a selective liberal arts college.
\newblock \emph{Research in Higher Education}, 37(2): 221--240.

\bibitem[{Celebi(2023)}]{celebi2023diversity}
Celebi, O. 2023.
\newblock Diversity preferences, affirmative action and choice rules.
\newblock \emph{arXiv preprint arXiv:2310.14442}.

\bibitem[{Chernoff(1954)}]{chernoff1954rational}
Chernoff, H. 1954.
\newblock Rational selection of decision functions.
\newblock \emph{Econometrica: journal of the Econometric Society}, 422--443.

\bibitem[{Creel and Hellman(2022)}]{creel2022algorithmic}
Creel, K.; and Hellman, D. 2022.
\newblock The algorithmic leviathan: Arbitrariness, fairness, and opportunity in algorithmic decision-making systems.
\newblock \emph{Canadian Journal of Philosophy}, 52(1): 26--43.

\bibitem[{Deng, Panigrahi, and Waggoner(2017)}]{deng2017complexity}
Deng, Y.; Panigrahi, D.; and Waggoner, B. 2017.
\newblock The complexity of stable matchings under substitutable preferences.
\newblock In \emph{Proceedings of the AAAI Conference on Artificial Intelligence}, volume~31, 480--486.

\bibitem[{Dong et~al.(2025)Dong, Schein, Wang, and Garg}]{dongdiscretization2025}
Dong, E.; Schein, A.; Wang, Y.; and Garg, N. 2025.
\newblock Addressing discretization-induced bias in demographic prediction.
\newblock \emph{PNAS Nexus}, 4(2): pgaf027.

\bibitem[{Echenique and Yenmez(2015)}]{echenique2015control}
Echenique, F.; and Yenmez, M.~B. 2015.
\newblock How to control controlled school choice.
\newblock \emph{American Economic Review}, 105(8): 2679--2694.

\bibitem[{Gardner, Popovic, and Schmidt(2023)}]{gardner2023benchmarking}
Gardner, J.; Popovic, Z.; and Schmidt, L. 2023.
\newblock Benchmarking distribution shift in tabular data with tableshift.
\newblock \emph{Advances in Neural Information Processing Systems}, 36: 53385--53432.

\bibitem[{Gray et~al.(2022)Gray, McCotter, Berry, Buschbacher, and Kelson}]{sffaamicus}
Gray, C.~B.; McCotter, R.~T.; Berry, J.; Buschbacher, M.; and Kelson, J.~M. 2022.
\newblock Brief of Economists as Amici Curiae in Support of Petitioner.
\newblock Students for Fair Admissions, Inc. v. President and Fellows of Harvard College and SFFA v. University of North Carolina (UNC), Supreme Court of the United States, 2023.

\bibitem[{Gupta, Sawhney, and Roth(2016)}]{gupta2016will}
Gupta, N.; Sawhney, A.; and Roth, D. 2016.
\newblock Will I get in? modeling the graduate admission process for American universities.
\newblock In \emph{2016 IEEE 16th international conference on data mining workshops (ICDMW)}, 631--638. IEEE.

\bibitem[{Jain, Creel, and Wilson(2024)}]{jain2024scarce}
Jain, S.; Creel, K.; and Wilson, A. 2024.
\newblock Scarce Resource Allocations That Rely On Machine Learning Should Be Randomized.
\newblock \emph{arXiv preprint arXiv:2404.08592}.

\bibitem[{Jain et~al.(2024)Jain, Suriyakumar, Creel, and Wilson}]{jain2023algorithmic}
Jain, S.; Suriyakumar, V.; Creel, K.; and Wilson, A. 2024.
\newblock Algorithmic pluralism: A structural approach to equal opportunity.
\newblock In \emph{Proceedings of the 2024 ACM Conference on Fairness, Accountability, and Transparency}, 197--206.

\bibitem[{Kalai, Rubinstein, and Spiegler(2002)}]{kalai2002rationalizing}
Kalai, G.; Rubinstein, A.; and Spiegler, R. 2002.
\newblock Rationalizing choice functions by multiple rationales.
\newblock \emph{Econometrica}, 70(6): 2481--2488.

\bibitem[{Kaur, Kiciman, and Sharma(2023)}]{kaurmodeling}
Kaur, J.~N.; Kiciman, E.; and Sharma, A. 2023.
\newblock Modeling the Data-Generating Process is Necessary for Out-of-Distribution Generalization.
\newblock In \emph{The Eleventh International Conference on Learning Representations}.

\bibitem[{Kiaghadi and Hoseinpour(2023)}]{kiaghadi2023university}
Kiaghadi, M.; and Hoseinpour, P. 2023.
\newblock University admission process: a prescriptive analytics approach.
\newblock \emph{Artificial Intelligence Review}, 56(1): 233--256.

\bibitem[{Kim et~al.(2025)Kim, Garg, Peng, and Garg}]{kim2025correlated}
Kim, E.~M.; Garg, A.; Peng, K.; and Garg, N. 2025.
\newblock Correlated Errors in Large Language Models.
\newblock In \emph{Forty-second International Conference on Machine Learning}.

\bibitem[{Kleinberg and Raghavan(2021)}]{kleinberg2021algorithmic}
Kleinberg, J.; and Raghavan, M. 2021.
\newblock Algorithmic monoculture and social welfare.
\newblock \emph{Proceedings of the National Academy of Sciences}, 118(22): e2018340118.

\bibitem[{Koh et~al.(2021)Koh, Sagawa, Marklund, Xie, Zhang, Balsubramani, Hu, Yasunaga, Phillips, Gao et~al.}]{koh2021wilds}
Koh, P.~W.; Sagawa, S.; Marklund, H.; Xie, S.~M.; Zhang, M.; Balsubramani, A.; Hu, W.; Yasunaga, M.; Phillips, R.~L.; Gao, I.; et~al. 2021.
\newblock Wilds: A benchmark of in-the-wild distribution shifts.
\newblock In \emph{International conference on machine learning}, 5637--5664. PMLR.

\bibitem[{Lee, Kizilcec, and Joachims(2023)}]{lee2023evaluating}
Lee, H.; Kizilcec, R.~F.; and Joachims, T. 2023.
\newblock Evaluating a learned admission-prediction model as a replacement for standardized tests in college admissions.
\newblock In \emph{Proceedings of the Tenth ACM Conference on Learning@ Scale}, 195--203.

\bibitem[{Lee et~al.(2024)Lee, Harvey, Zhou, Garg, Joachims, and Kizilcec}]{lee2024ending}
Lee, J.; Harvey, E.; Zhou, J.; Garg, N.; Joachims, T.; and Kizilcec, R.~F. 2024.
\newblock Ending affirmative action harms diversity without improving academic merit.
\newblock In \emph{Proceedings of the 4th ACM Conference on Equity and Access in Algorithms, Mechanisms, and Optimization}, 1--17.

\bibitem[{Lee et~al.(2023)Lee, Thymes, Zhou, Joachims, and Kizilcec}]{lee2023augmenting}
Lee, J.; Thymes, B.; Zhou, J.; Joachims, T.; and Kizilcec, R.~F. 2023.
\newblock Augmenting holistic review in university admission using natural language processing for essays and recommendation letters.
\newblock \emph{arXiv preprint arXiv:2306.17575}.

\bibitem[{Liu et~al.(2023)Liu, Wang, Britton, and Abebe}]{liu2023reimagining}
Liu, L.~T.; Wang, S.; Britton, T.; and Abebe, R. 2023.
\newblock Reimagining the machine learning life cycle to improve educational outcomes of students.
\newblock \emph{Proceedings of the National Academy of Sciences}, 120(9): e2204781120.

\bibitem[{Lux et~al.(2016)Lux, Pittman, Shende, and Shende}]{lux2016applications}
Lux, T.; Pittman, R.; Shende, M.; and Shende, A. 2016.
\newblock Applications of supervised learning techniques on undergraduate admissions data.
\newblock In \emph{Proceedings of the ACM International Conference on Computing Frontiers}, 412--417.

\bibitem[{Moore(1998)}]{moore1998expert}
Moore, J.~S. 1998.
\newblock An expert system approach to graduate school admission decisions and academic performance prediction.
\newblock \emph{Omega}, 26(5): 659--670.

\bibitem[{Moulin(1985)}]{moulin1985choice}
Moulin, H. 1985.
\newblock Choice functions over a finite set: a summary.
\newblock \emph{Social Choice and Welfare}, 2(2): 147--160.

\bibitem[{Neda and Gago-Masague(2022)}]{neda2022feasibility}
Neda, B.~M.; and Gago-Masague, S. 2022.
\newblock Feasibility of machine learning support for holistic review of undergraduate applications.
\newblock In \emph{2022 International Conference on Applied Artificial Intelligence (ICAPAI)}, 1--6. IEEE.

\bibitem[{{NYC DOE}(2022)}]{nycdoe2022hsdata}
{NYC DOE}. 2022.
\newblock Fall 2022 {High School} {Data}.
\newblock InfoHub Directory Data.

\bibitem[{Peng and Garg(2024{\natexlab{a}})}]{peng2024monoculture}
Peng, K.; and Garg, N. 2024{\natexlab{a}}.
\newblock Monoculture in Matching Markets.
\newblock In \emph{The Thirty-eighth Annual Conference on Neural Information Processing Systems}.

\bibitem[{Peng and Garg(2024{\natexlab{b}})}]{peng2024wisdomfoolishnessnoisymatching}
Peng, K.; and Garg, N. 2024{\natexlab{b}}.
\newblock Wisdom and Foolishness of Noisy Matching Markets.
\newblock In \emph{Proceedings of the 25th ACM Conference on Economics and Computation}, 675--675.

\bibitem[{Peng et~al.(2025)Peng, Ryu, Kleinberg, Tardos, and Garg}]{peng2025undermatching}
Peng, K.; Ryu, E.; Kleinberg, J.; Tardos, E.; and Garg, N. 2025.
\newblock Undermatching in New York City School Choice: Application Behavior and the Potential of Personalized Feedback.
\newblock \emph{arXiv preprint}.

\bibitem[{Perdomo et~al.(2025)Perdomo, Britton, Hardt, and Abebe}]{perdomo2025difficult}
Perdomo, J.~C.; Britton, T.; Hardt, M.; and Abebe, R. 2025.
\newblock Difficult lessons on social prediction from wisconsin public schools.
\newblock In \emph{Proceedings of the 2025 ACM Conference on Fairness, Accountability, and Transparency}, 2682--2704.

\bibitem[{Raji et~al.(2022)Raji, Kumar, Horowitz, and Selbst}]{raji2022fallacy}
Raji, I.~D.; Kumar, I.~E.; Horowitz, A.; and Selbst, A. 2022.
\newblock The fallacy of AI functionality.
\newblock In \emph{Proceedings of the 2022 ACM Conference on Fairness, Accountability, and Transparency}, 959--972.

\bibitem[{Rauls(2021)}]{collegevine}
Rauls, L. 2021.
\newblock Is CollegeVine’s Chancing Engine Actually Accurate?

\bibitem[{Rees and Ryder(2023)}]{rees2023machine}
Rees, C.~A.; and Ryder, H.~F. 2023.
\newblock Machine learning for the prediction of ranked applicants and matriculants to an internal medicine residency program.
\newblock \emph{Teaching and Learning in Medicine}, 35(3): 277--286.

\bibitem[{Roth(1984)}]{roth1984stability}
Roth, A.~E. 1984.
\newblock Stability and polarization of interests in job matching.
\newblock \emph{Econometrica: Journal of the Econometric Society}, 47--57.

\bibitem[{Sen(1971)}]{sen1971choice}
Sen, A.~K. 1971.
\newblock Choice functions and revealed preference.
\newblock \emph{The Review of Economic Studies}, 38(3): 307--317.

\bibitem[{Shen-Berro(2024)}]{ShenBerro2024HighSchoolAdmissions}
Shen-Berro, J. 2024.
\newblock NYC Education Department developing tool to help applicants gauge high school admissions chances.
\newblock \emph{Chalkbeat New York}.
\newblock Accessed: 2024-10-20.

\bibitem[{Sirolly, Kanoria, and Ma(2024)}]{sirolly2024impact}
Sirolly, A.; Kanoria, Y.; and Ma, H. 2024.
\newblock The Impact of Race-Blind and Test-Optional Admissions on Racial Diversity and Merit.
\newblock In \emph{Proceedings of the 25th ACM Conference on Economics and Computation}, 308--308.

\bibitem[{Staudaher, Lee, and Soleimani(2020)}]{staudaher2020predicting}
Staudaher, S.; Lee, J.; and Soleimani, F. 2020.
\newblock Predicting Applicant Admission Status for Georgia Tech's Online Master's in Analytics Program.
\newblock In \emph{Proceedings of the Seventh ACM Conference on Learning@ Scale}, 309--312.

\bibitem[{Toups et~al.(2023)Toups, Bommasani, Creel, Bana, Jurafsky, and Liang}]{toups2023ecosystem}
Toups, C.; Bommasani, R.; Creel, K.; Bana, S.; Jurafsky, D.; and Liang, P.~S. 2023.
\newblock Ecosystem-level analysis of deployed machine learning reveals homogeneous outcomes.
\newblock \emph{Advances in Neural Information Processing Systems}, 36: 51178--51201.

\bibitem[{Wang et~al.(2024)Wang, Kapoor, Barocas, and Narayanan}]{wang2024against}
Wang, A.; Kapoor, S.; Barocas, S.; and Narayanan, A. 2024.
\newblock Against predictive optimization: On the legitimacy of decision-making algorithms that optimize predictive accuracy.
\newblock \emph{ACM Journal on Responsible Computing}, 1(1): 1--45.

\bibitem[{Waters and Miikkulainen(2014)}]{waters2014grade}
Waters, A.; and Miikkulainen, R. 2014.
\newblock Grade: Machine learning support for graduate admissions.
\newblock \emph{Ai Magazine}, 35(1): 64--64.

\bibitem[{Yao et~al.(2022)Yao, Choi, Cao, Lee, Koh, and Finn}]{yao2022wild}
Yao, H.; Choi, C.; Cao, B.; Lee, Y.; Koh, P. W.~W.; and Finn, C. 2022.
\newblock Wild-time: A benchmark of in-the-wild distribution shift over time.
\newblock \emph{Advances in Neural Information Processing Systems}, 35: 10309--10324.

\bibitem[{Yokoi(2019)}]{yokoi2019matroidal}
Yokoi, Y. 2019.
\newblock Matroidal choice functions.
\newblock \emph{SIAM Journal on Discrete Mathematics}, 33(3): 1712--1724.

\end{thebibliography}
